\documentclass[a4paper,UKenglish,11pt]{article}

\usepackage{typearea}
\typearea{15}

\usepackage[dvipsnames]{xcolor}
\usepackage[shortlabels]{enumitem}	

\usepackage{etex}               
\usepackage{amsmath}            
\usepackage{amssymb}            
\usepackage{amsthm}             
\usepackage{pst-all}
\usepackage{pstricks-add}
\usepackage{pst-3dplot}
\usepackage{float}
\usepackage{floatflt}
\usepackage{booktabs,framed}

\newtheorem{theorem}{Theorem}
\newtheorem{lemma}[theorem]{Lemma}
\newtheorem{corollary}[theorem]{Corollary}

\newtheorem{proposition}[theorem]{Proposition}

\theoremstyle{definition}
\newtheorem{definition}[theorem]{Definition}

\newtheorem{remark}[theorem]{Remark}

\usepackage{hyperref}

\renewcommand{\epsilon}{\varepsilon}

\usepackage{xspace,cite}

\usepackage[ruled,vlined,linesnumbered]{algorithm2e}

\newcommand{\classPLS}{{\sf PLS}}

   \def\1{{\mathbbm 1}}

   \def\R{{\mathbb R}}

\DeclareMathOperator{\rk}{\text{rk}}

   \def\cB{{\mathcal B}}

   \def\cM{{\mathcal M}}

   \def\cP{{\mathcal P}}

   \def\cS{{\mathcal S}}

\usepackage{graphicx}        
\usepackage{multicol}        
\usepackage[bottom]{footmisc}
\usepackage[textsize=tiny,textwidth=2cm,shadow]{todonotes}

\makeatletter
\let\original@algocf@latexcaption\algocf@latexcaption
\long\def\algocf@latexcaption#1[#2]{%
  \@ifundefined{NR@gettitle}{%
    \def\@currentlabelname{#2}%
  }{%
    \NR@gettitle{#2}%
  }%
  \original@algocf@latexcaption{#1}[{#2}]%
}
\makeatother

\makeatletter
\newcommand{\mylabel}[2]{#2\def\@currentlabel{#2}\label{#1}}
\makeatother

\newcommand{\Kante}[3]{\ncline{-}{#1}{#2}\ncput{\colorbox{almostwhite}{$#3$}}}
\newcommand{\Kantedashed}[3]{\ncline[linestyle=dashed]{-}{#1}{#2}\ncput{\colorbox{almostwhite}{$#3$}}}
\newcommand{\Knoten}[3]{\cnode*(#1,#2){2.2pt}{#3}}

\newgray{almostwhite}{0.9}

\makeindex             

\title{Efficient Black-Box Reductions for Separable Cost Sharing}

\author{%
Tobias Harks\thanks{Universit\"at Augsburg, Institut f\"ur Mathematik, Germany. \texttt{tobias.harks@math.uni-augsburg.de}}%
\and%
Martin Hoefer\thanks{Goethe University Frankfurt, Institute of Computer Science, Germany. \texttt{mhoefer@cs.uni-frankfurt.de}}%
\and%
Anja Huber\thanks{Universit\"at Augsburg, Institut f\"ur Mathematik, Germany. \texttt{anja.huber@math.uni-augsburg.de}}%
\and%
Manuel Surek\thanks{Universit\"at Augsburg, Institut f\"ur Mathematik, Germany. \texttt{manuel.surek@math.uni-augsburg.de}}}

\date{}

\begin{document}

\maketitle

\begin{abstract}
In \emph{cost sharing games with delays,} a set of agents jointly allocates a finite subset of resources. Each resource has a fixed cost that has to be shared by the players, and each agent has a non-shareable player-specific delay for each resource. A prominent example is uncapacitated facility location (UFL), where facilities need to be opened (at a shareable cost) and clients want to connect to opened facilities. Each client pays a cost share and his non-shareable physical connection cost. Given any profile of subsets allocated by the agents, a \emph{separable cost sharing protocol} determines cost shares that satisfy budget balance on every resource and separability over the resources. Moreover, a separable protocol guarantees existence of pure Nash equilibria in the induced strategic game for the agents.

In this paper, we study separable cost sharing protocols in several general combinatorial domains. We provide black-box reductions to reduce the design of a separable cost-sharing protocol to the design of an approximation algorithm for the underlying cost minimization problem. In this way, we obtain new separable cost-sharing protocols in games based on arbitrary player-specific matroids, single-source connection games without delays, and connection games on $n$-series-parallel graphs with delays. All these reductions are efficiently computable -- given an initial allocation profile, we obtain a cheaper profile and separable cost shares turning the profile into a pure Nash equilibrium. Hence, in these domains any approximation algorithm can be used to obtain a separable cost sharing protocol with a price of stability bounded by the approximation factor.

\end{abstract}

\clearpage

\section{Introduction}
\label{sec:intro}

Cost sharing is a fundamental task in networks with strategic agents and has attracted a large amount of interest in algorithmic game theory. Traditionally, cost sharing has been studied in a cooperative sense, i.e., in the form of cooperative games or mechanism design. Many of these approaches treat cost in a \emph{non-separable} way and return a single, global cost share for each agent. In contrast, when agents jointly design a resource infrastructure in large networks, it is much more desirable to provide algorithms and protocols for \emph{separable} cost sharing that specify which agent needs to pay how much to each resource. Here the natural approach are strategic cost sharing games with $n$ players that allocate subsets of $m$ resources. Each resource generates a cost depending on the subset of players allocating it. A protocol determines a cost share for each resource and each player using it. In addition to separability, there are further natural desiderata for such protocols, such as budget-balance (distribute exactly the arising cost of each resource) and existence of a pure Nash equilibrium (PNE), i.e., allow the resulting game to stabilize.

Perhaps the most prominent such protocol is the fair-share protocol, in which the cost of each resource is allocated in equal shares to the players using it. This approach has been studied intensively (see our discussion below), but there are several significant drawbacks. It can be \classPLS-hard to find~\cite{Syrgkanis10} a PNE, even in connection games on undirected networks. The price of stability (PoS), i.e., the total cost of the best Nash equilibrium compared to the cost of the optimal allocation, can be as large as $\Omega(\log n)$~\cite{AnshelevichDKRTW08, ChenRV10}, even though much better solutions can often be found in polynomial time.

In this paper, we study a slight generalization of cost sharing games, where every resource has a shareable cost component and a non-shareable player-specific delay component. The shareable cost needs to be shared by the players using it, the non-shareable player-specific delay represents, e.g., a physical delay and is thus unavoidable. This setting arises in several relevant scenarios, such as uncapacitated facility location (UFL)~\cite{HarksF14}. Here players share the monetary cost of opened facilities but additionally experience delays measured by the distance to the closest open facility. Another important example appears in network design, where players jointly buy edges of a graph to connect their terminals. Besides the monetary cost for buying edges, each player experiences player-specific delays on the chosen paths. In such a distributed network environment, it is not clear a priori if an optimal solution can be stable -- i.e., if the shareable costs can be distributed among the players in a separable way so that players do not want to deviate from it. This question leads directly to the design of protocols that distribute the costs in order to induce stable and good-quality solutions of the resulting strategic game. 

Our results are three polynomial-time black-box reductions for the price of stability of separable cost sharing protocols in combinatorial resource allocation problems. Our domains represent broad generalizations of UFL -- arbitrary, player-specific matroids, single-source connection games without delays, and connection games on undirected $n$-series-parallel graphs with delays. In each of these domains, we take as input an arbitrary profile and efficiently turn it into a cheaper profile and a sharing of the shareable costs such that it is a Nash equilibrium. Our protocols are polynomial-time in several ways. Firstly, the games we study are succinctly represented. In matroids, we assume that strategies are represented implicitly via an independence oracle. For connection games on graphs, the strategies of each player are a set of paths, which is implicitly specified by terminal vertices of the player and the graph structure. The cost sharing protocol is represented by a strategy profile $S$ and a sharing of the shareable costs arising in $S$ on each resource. While in principle the protocol must specify a sharing of the costs for all of the other (possibly exponentially many) strategy profiles, one can do so implicitly by a simple lexicographic assignment rule. It guarantees that the profile $S$ becomes a PNE. As such, starting from an arbitrary initial profile $S'$, we can give in polynomial time the Nash equilibrium profile $S$, the cost shares for $S$, and the assignment rule for cost shares in the other profiles. Hence, if $S'$ is polynomial-time computable, then both protocol and Nash equilibrium $S$ are both polynomial-time computable and polynomial-space representable.

\subsection{Our Results}

We present several new polynomial-time black-box reductions for separable cost sharing protocols with small price of stability (PoS). We study three domains that represent broad generalizations of the uncapacitated facility location problem. In each domain, we devise an efficient black-box reduction that takes as input an arbitrary strategy profile and computes a new profile of lower cost together with a separable cost sharing protocol inducing the cheaper profile as a PNE. Thus, \emph{any} polynomial-time $\alpha$-approximation of the social cost can be turned into a separable cost sharing protocol with PoS at most $\alpha$.

\paragraph{Matroidal Set Systems.} 
In Section~\ref{sec:matroids} we provide a black-box reduction for matroidal set systems. Our results even apply to the broad class of \emph{subadditive} cost functions that include fixed costs and discrete concave costs even with weighted players as a special case. Here we assume access to a value oracle for the subadditive cost function for each resource. Matroidal set systems with player-specific delays include uncapacitated facility location as a special case, since these correspond to matroid games, where each player has a uniform rank $1$ matroid. For \emph{metric} UFL, there is for instance a $1.488$-approximation algorithm~\cite{Li:2013} using ideas of a previous $1.5$-approximation algorithm~\cite{Byrka10}. This leads to a separable cost sharing protocol with PoS of $1.488$. Also, the existing hardness results for UFL carry over to the design of separable cost sharing protocols, and for metric UFL there is a lower bound of $1.46$~\cite{GuhaK99}.

\paragraph{Connection Games with Fixed Costs.}
In Section~\ref{sec:fixed} we consider cost sharing games on graphs, where the set systems correspond to paths connecting a player-specific source with a player-specific terminal. The underlying optimization problem is Steiner forest. For multi-terminal connection games without delays, we observe that a simple greedy algorithm for the underlying Steiner forest problem combined with the idea of Prim-Sharing~\cite{ChenRV10} yields a separable protocol in polynomial time. Since the greedy algorithm has recently been shown to provide a constant-factor approximation~\cite{GuptaK15}, the protocol yields a constant PoS.

For single-source multi-terminal connection games we again provide a polynomial-time black-box reduction. Our result improves significantly over the existing Prim-Sharing~\cite{ChenRV10} with a price of stability of 2. We obtain separable protocols based on any approximation algorithm for Steiner tree, such as, e.g., the classic 1.55-approximation algorithm~\cite{RobinsZ05}, or the celebrated recent 1.39-approximation algorithm~\cite{ByrkaGRS13}. Our black-box reduction continues to hold even for directed graphs, where we can use any algorithm for the Directed Steiner Tree problem~\cite{CharikarCCDGG99}, or games based on the (directed or undirected) Group Steiner Tree problem~\cite{GargKR00,ChekuriP05}. Similarly, all lower bounds on approximation hardness translate to the price of stability of polynomial-time computable separable protocols. 

\paragraph{Connection Games with Delays.} 
Finally, in Section~\ref{sec:nSepa} we study multi-terminal connection games with delays and fixed costs. For directed graphs, an optimal Steiner forest is not enforceable by a separable cost sharing protocol, even for two players~\cite{ChenRV10}. Very recently, a similar result was shown even for two-player games on undirected graphs~\cite{harks2017}. Thus, for general graphs, we cannot expect separable protocols with optimal or close-to-optimal equilibria, or (efficient) black-box reductions. We introduce a class of so-called $n$-series-parallel graphs, which allows to obtain a black-box reduction in polynomial time. The transformation directly implies that the $n$-series-parallel graphs always admit a separable cost sharing protocol inducing an optimal Steiner forest as an equilibrium. 

The reduction also applies to discrete-concave cost functions and player-specific delays, however, we do not know if polynomial running time is guaranteed. $n$-series-parallel graphs have treewidth at most 2, thus, for fixed edge costs and no delays, it is possible to compute efficiently even an optimal Steiner forest~\cite{bateni2011}. Hence, in this case we obtain a separable protocol with PoS of 1 in polynomial time. We finally demonstrate that the specific setting of $n$-series-parallel graphs is in some sense necessary: Even for generalized series-parallel graphs we give a counterexample showing that a black-box reduction is impossible to achieve.

\subsection{Preliminaries and Related Work}
Cooperative cost sharing games have been studied over the last decades for a variety of combinatorial optimization problems, such as minimum spanning tree~\cite{Bird76}, Steiner tree~\cite{Megiddo78, GranotH81, GranotM98, Tamir91}, facility location~\cite{GoemansS04}, vertex cover~\cite{DengIN99}, and many more. Cooperative cost sharing games have interesting implications for (group-)strategyproof cost sharing mechanisms~\cite{MoulinShenker, JainV01, KonemannLSZ08, PalT03}. For Bayesian cost-sharing mechanisms there even exist efficient black-box reductions from algorithm to mechanism design~\cite{GeorgiouS12}. A major difference to our work is that cooperative cost sharing is \emph{not separable}.

The most prominent example of a \emph{separable cost sharing protocol} is the \emph{fair-share} protocol, in which the cost of each resource is divided in equal shares among the players that allocate it. This protocol is also anonymous, and it implies that the resulting game is a congestion game~\cite{Rosenthal73}. It guarantees the smallest price of stability within a class of anonymous protocols~\cite{ChenRV10}. The fair-share protocol has attracted a serious amount of research interest over the last decade~\cite{AnshelevichDKRTW08, AndelmanFM09, Bilo2010, HansenT09}, especially the notorious open problem of a constant price of stability for connection games in undirected graphs~\cite{FiatKLOS06, Li09, LeeL13, BiloFM13, DisserFKM15}. However, as a significant drawback, outside of the domain of undirected connection games the price of stability is often as large as $\Omega(\log n)$. Moreover, computing a PNE is \classPLS-hard, even for undirected connection games~\cite{Syrgkanis10}.

More general separable protocols have been studied mostly in terms of the price of anarchy, e.g., for scheduling (or matroid games)~\cite{AvniT16,CGV17,ChristodoulouGS17,HarksF13,Feldman12} or single-source network design with~\cite{CS16,LS16} and without uncertainty~\cite{ChenRV10}. The best result here is a price of anarchy (and stability) of 2 via Prim-Sharing~\cite{ChenRV10}, a protocol inspired by Prim's MST algorithm. A protocol with logarithmic price of stability was shown for capacitated UFL games~\cite{HarksF14}.

We note here that separable protocols with low PoS can be obtained using results for cost sharing games with so-called \textit{arbitrary sharing}. In cost sharing games with arbitrary sharing, each agent $i \in N$ specifies as strategy a set $S_i$ of allocated resources and a payment $p_{i,e}$ for every resource $e \in E$. A resource $e \in E$ is \emph{bought} if the total payments exceed the costs $\sum_i p_{i,e} \ge c_e(S)$. The private cost of player $i$ is the sum of all payments $\sum_e p_{i,e}$ if all resources $e \in S_i$ are bought, and $\infty$ otherwise. Note that for games with fixed costs $c_e(S) = c_e$, one usually drops the explicit allocation $S_i$ from the strategy of a player. Instead, each player $i \in N$ simply specifies strategic payments $p_{i,e}$ for each $e \in E$. Then the private cost of player $i$ is $\sum_e p_{i,e}$ if payments suffice to buy at least one feasible set in $\cS_i$, and $\infty$ otherwise. The following proposition is an interesting, straightforward insight. It has been observed before in the special case of single-source connection games~\cite[Proposition 6.5]{ChenRV10}.

\begin{proposition}
   If for a cost sharing model, the non-cooperative game with arbitrary sharing has a pure Nash equilibrium, then there is a separable cost sharing protocol with the same pure Nash equilibrium.
\end{proposition}

\begin{proof}
It is easy to see that in a PNE $(S,p)$ for a game with arbitrary sharing, every player $i \in N$ contributes only to resources from one feasible set $S_i \in \cS_i$. Moreover, the cost of every resource is exactly paid for. Finally, if a player $i$ deviates to a different feasible set $S_i'$, then for each $e \in S_i' \setminus S_i$ she only needs to contribute the marginal costs that arise due to her presence. In particular, for fixed costs, she can use all resources bought by others for free. 

Hence, given a PNE $(S,p)$ for the game with arbitrary sharing, we obtain a basic and separable protocol $\Xi$ as follows. If in profile $S'$ a resource $e$ is allocated by the set $N_e(S)$, then we assign $\xi_{e,i}(S') = p_{i,e}$ for every $i \in N_e(S)$. For a profile $S'$, in which at least one other player $i \in N_e(S') \setminus N_e(S)$ allocates $e$, we pick one of these players $i$, and she has to pay the full cost $\xi_{i,e}(S') = c_e(S')$. If players from a strict subset $N_e(S') \subset N_e(S)$ allocate $e$ in $S'$, we can use an arbitrary budget-balanced sharing of $c_e(S')$. It is straightforward to verify that such a protocol is basic and separable, and the state $S$ is a PNE.
\end{proof}

This implies existence of separable protocols with optimal PNE and price of stability 1 for a variety of classes of games, including matroid games with uniform discrete-concave costs~\cite{HarksP14}, uncapacitated facility location with fixed~\cite{CardinalH10} and discrete-concave costs~\cite{Hoefer11}, connection games (single-source~\cite{Hoefer09, AnshelevichDTW08} and other classes~\cite{AnshelevichC09a, AnshelevichC09b, Hoefer13}) with fixed costs, and more. However, the large majority of these results are \emph{inefficient}, i.e., there is no polynomial-time algorithm that computes the required optimal equilibrium. 

Alternatively, one may resort to approximate equilibria in games with arbitrary sharing that are efficiently computable. The most prominent technique works via reducing costs by an additive value $\epsilon$ to ensure polynomial running time (put forward for single-source connection games in~\cite{AnshelevichDTW08} and used in much of the follow-up work~\cite{Hoefer09,AnshelevichC09a,AnshelevichC09b,CardinalH10}). This approach \emph{does not translate} to separable protocols, since a player must eventually contribute to \emph{all resources}. This is impossible for the model we consider here.

\section{Separable Cost Sharing Protocols}

We are given a finite set $N$ of players and a finite set $E$ of resources. Each player $i\in N$ is associated with a predefined family of subsets $\cS_i\subseteq 2^E$ from which player $i$ needs to pick at least one. The space of strategy profiles is denoted by  $\mathcal{S}:= \times_{i \in N} \mathcal{S}_i$. For $S\in\mathcal{S}$ we denote by $N_e(S)=\{i\in N: e \in S_i\}$ the set of players that allocate resource $e$. Every resource $e \in E$ has a fixed cost $c_e\geq 0, e\in E$ that is assumed to be \emph{shareable} by the players. In addition to the shareable costs, there are \emph{player-specific constant costs} $d_{i,e}\geq 0, i\in N, e\in E$ that are not shareable. If player $i$ chooses subset $S_i$, then the player-specific costs $\sum_{e\in S_i}d_{i,e}$ must be 
paid completely by player $i$. The total cost of a profile $S$ is defined as $C(S)=\sum_{e\in \cup_{i\in N} S_i} c_e+\sum_{i\in N}\sum_{e\in S_i} d_{i,e}.$

A \emph{cost sharing protocol} $\Xi$ assigns cost share functions $\xi_{i,e} : \cS \rightarrow \R_{\geq 0}$ for all $i\in N$ and $e\in E$ and thus induces the strategic game $(N,\cS,\xi)$. For a player $i$, her total private cost of strategy $S_i$ in profile $S$ is $\xi_i(S):=\sum_{e\in S_i}{(\xi_{i,e}(S)+d_{i,e})}$. We assume that every player picks a strategy in order to minimize her private cost. A prominent solution concept in non-cooperative game theory are pure Nash equilibria. Using standard notation in game theory, for a strategy profile $S \in \cS$ we denote by $(S'_i,S_{-i}) := (S_1,\dots,S_{i-1},S'_i,S_{i+1},\dots,S_n) \in \cS$ the profile that arises if only player $i$ deviates to strategy~$S'_i\in \cS_i$. A profile is a \emph{pure Nash equilibrium (PNE)} if for all $i \in N$ it holds $\xi_i(S) \leq \xi_i(S'_i,S_{-i})$ for all $S'_i\in \cS_i$.

In order to be practically relevant, cost sharing protocols need to satisfy several desiderata. In this regard, \emph{separable} cost sharing protocols
are defined as follows~\cite{ChenRV10}.

\begin{definition}[Cost Sharing Protocols and Enforceability]
  A cost sharing protocol $\Xi$ is 
  \begin{enumerate}
    \item {\em stable} if it induces only games that admit at least one pure Nash equilibrium.
    \item \emph{budget balanced,} if for all $e\in E$ with $N_e(S)\neq \emptyset$
    \begin{align*}	
      c_e &= \sum_{i\in N_e(S)}{\xi_{i,e}(S)} 
	  \text{ and } \xi_{i,e}(S) =0 \text{ for all $i\not\in N_e(S)$.}
	\end{align*}
	\item {\em separable} if it is stable, budget-balanced and induces only games for which in any two profiles $S,S' \in \cS$ for every resource $ e \in E$,
	\[N_e(S)=N_e(S')\Rightarrow \xi_{i,e}(S) = \xi_{i,e}(S') \text{ for all } i \in N_e(S). \]
	\item {\em polynomial time computable,} if the cost sharing functions $\xi$ can be computed in polynomial time in the encoding length of the cost sharing game.
  \end{enumerate}
  We say that a strategy profile $S$ is \emph{enforceable}, if there is
  a separable protocol inducing $S$ as a pure Nash equilibrium.
\end{definition}
Separability means that for any two profiles $S,S'$ the cost shares on $e$ are the same if the set of players using $e$ remains unchanged. Still, separable protocols can assign cost share functions that are specifically tailored to a given congestion model, for example based on an optimal profile. In this paper, we are additionally interested in \emph{polynomial-time computable protocols} that we introduce here.

\section{Matroid Games}
\label{sec:matroids}

In this section, we consider matroid games. As usual in matroid theory, we will write $\cB_i$ instead of $\cS_i$, and $\cB$ instead of $\cS$,
when considering matroid games. The tupel $\mathcal{M} = (N, E, \cB, (c_e)_{e \in E}, (d_{i,e})_{e\in E,i\in N})$ is called a \emph{matroid game} if $E=\bigcup_{i\in N} E_i$, and each set system $\cB_i\subseteq 2^{E_i}$ forms the base set of some matroid $\cM_i=(E_i, \cB_i)$. While seemingly abstract, the class includes several prominent application domains, such as UFL games. In a UFL game, the resources are facilities (e.g. common transport hubs) and the players incur delay $d_{i,e}$ in addition to their cost shares for opening used facilities. Every player $i$ chooses exactly one resource, that is $|B_i|=1$ for all $B_i\in \cB_i$ and $i \in N$ and hence $\cB_i$ corresponds to a \emph{uniform matroid of rank one}. Recall that every base $B$ of a matroid $\cM_i=(E_i, \cB_i)$ has the same cardinality which we denote with $\rk_i$ (the rank of $\cM_i$).

In the following, instead of fixed costs on the resource, we allow
for general \emph{subadditive} cost functions $c_e:2^N\rightarrow \R_+, e\in E$. $c_e$ is called \emph{subadditive}, if it satisfies (1) $c_e(S)\leq c_e(T)$ for all $S\subseteq T\subseteq N$, and (2) $c_e(S+\{i\})\leq c_e(S)+c_e(\{i\})$ for all $S\subset N, i\in N$. Note that subadditive functions include fixed costs and discrete concave costs as a special case including the possibility of weighted demands as in weighted congestion games.

Let us denote the cost of the cheapest alternative of player $i$ to resource $e$ for profile $B\in\cB$ by 
$ \Delta_i^e(B):=\min_{\substack{f\in E\\B_i+f-e\in\cB_i}}{\left(c_f(B_i+f-e,B_{-i})+d_{i,f}\right)}. $
Here we use the intuitive notation $c_e(B):=c_e(N_e(B))$.
We recapitulate a characterization of enforceable strategy profiles obtained
in~\cite{HarksF14}.\footnote{The original characterization in~\cite{HarksF14} was proven for weighted players and load-dependent non-decreasing cost functions but the proof also works for subadditive cost functions.}

\begin{lemma}\label{dech}
A collection of bases $B=(B_1,\dots, B_n)$ is enforceable by a separable protocol if and only if the following two properties are satisfied. Note that \eqref{transp} implies that each summand $\Delta^e_i(B)-d_{i,e}$ in \eqref{share} is nonnegative.
\begin{align}\label{transp}\tag{D1}
   d_{i,e}& \leq \Delta^e_i(B) \text{ for all }i\in N, e\in B_i\\
  \label{share}\tag{D2}
    c_e(B)&\leq \sum_{i\in N_e(B)}{\left(\Delta^e_i(B)-d_{i,e}\right)}  \text{ for all }e\in E.
\end{align}
\end{lemma}

\begin{algorithm}[t]
\DontPrintSemicolon
\SetKwInput{Input}{Input}
\SetKwInput{Output}{Output}
\Input{Congestion model $(N, \cB, c, d)$ and profile $B\in \cB$}
\Output{Enforceable profile $B'$ with $C(B')\leq C(B)$.}
\vspace{0.25cm}
Set $B' \leftarrow B$\;
\While{there is $e \in E$ that satisfies at least one of the following conditions:
\begin{minipage}[b]{0.48\linewidth}
\begin{equation}\label{eq:viol-delay} d_{i,e}> \bar \Delta^e_i(B') \text{ for some }i\in N_e(B')\end{equation}
\end{minipage}
or
\begin{minipage}[b]{0.48\linewidth}
\begin{equation}\label{eq:viol-cost}
   c_e(B')> \sum_{i\in N_e(B')}{\left(\bar \Delta^e_i(B')-d_{i,e}\right)}
\end{equation}
\end{minipage}
}{
\If{\eqref{eq:viol-delay} \text{holds true for some }$i\in N_e(B')$}{ 
Let $f_i\in \arg\min\limits_{\substack{f\in E\\B'_i+f-e\in\cB_i}}{\pi_i^f}$\;
Update $B'_i \leftarrow B'_i+f_i-e$\;\label{move1}
}
\ElseIf{\eqref{eq:viol-cost} holds true}{
\While{\eqref{eq:viol-cost} holds true on $e$ \label{alg:second-while}}{
Pick $i\in N_e(B')$ with $\pi_e^i > \bar \Delta^e_i(B')$\;\label{alg:pick}
Let $f_i\in \arg\min\limits_{\substack{f\in E\\B'_i+f-e\in\cB_i}}{\pi_i^f}$\;
Update $B'_i\leftarrow B'_i+f_i-e$\;\label{move2}
}
}}
\caption{Transforming any profile $B$ into an enforceable profile $B'$\label{alg:matroidTransform}}
\end{algorithm}

\begin{remark}
The characterization was used in~\cite{HarksF14} to prove that an optimal collection of bases is enforceable. This implies a PoS of $1$ for a separable cost sharing protocol that relies on the optimal profile. As such, the protocol is not efficiently computable (unless $P=NP$).
\end{remark}

In the following, we devise a black-box reduction in Algorithm~\ref{alg:matroidTransform}. It takes as input an arbitrary
collection of bases $B$ and transforms them \emph{in polynomial time} into an enforceable set of bases $B'$ of lower cost $C(B')\leq C(B)$. We define for each $i \in N, e\in E$ a \emph{virtual cost value} $\pi_i^e = c_{e}(\{i\}) + d_{i,e}$,
and for each $B\in \cB, i \in N, e\in E$ a \emph{virtual deviation cost}
$\bar \Delta^e_i(B) = \min_{\substack{f\in E\\B_i+f-e\in\cB_i}}{\pi_i^f}.$
The algorithm now iteratively checks whether \eqref{transp} and \eqref{share} from Lemma~\ref{dech} hold true (in fact it checks this condition for smaller values on the right hand side given by the virtual values), and if not, exchanges one element of some player. We show that the algorithm terminates with an enforceable profile after polynomially many steps.

\begin{theorem}
Let $B$ be a strategy profile for a matroid congestion model with subadditive costs. There is an enforceable strategy $B'$ with $C(B')\leq C(B)$ that can be computed in at most $n\cdot m\cdot \rk(\cB)$ iterations of the while-loop in Algorithm~\ref{alg:matroidTransform}, where $\rk(\cB) = \max_{i\in N}\rk_i$.
\end{theorem}

\begin{proof}
First, observe that if \eqref{transp} and \eqref{share} from Lemma~\ref{dech} hold true for smaller values $0\leq \bar \Delta^e_i (B)\leq \Delta^e_i(B), i\in N, e\in E$, then the profile $B$ is also enforceable. Hence, if the algorithm terminates, the resulting strategy profile $B'$ will be enforceable. 

To show that the algorithm is well-defined, we only need to check Line~\ref{alg:pick}. By
subadditivity we get 
$\sum_{i\in N_e(B')}c_e(\{i\}) \geq c_e(B').$
Thus, whenever
$c_e(B')> \sum_{i\in N_e(B')}\left( \bar \Delta^e_i(B')-d_{i,e}\right),$
there is an $i\in N_e(B')$ with
$c_e(\{i\})+d_{i,e}>\bar \Delta_i^e(B')$.

It is left to bound the running time. For this we consider player $i$ and the matroid bases $\cB_i$. We interpret a basis $B_i\in \cB_i$ as distributing exactly $\rk_i$  unit sized packets over the resources in $E$. This way, we can interpret the algorithm
as iteratively moving packets away from those resources $e\in E$ for which either \eqref{eq:viol-delay} or \eqref{eq:viol-cost}
holds true. We give each packet a unique ID $i_k, k=1,\dots, \rk_i$. For $B_i\in \cB_i$, let $e_{i_k}$ denote the resource on which packet $i_k$ is located. 
We now analyze the two types of packet movements during the execution of the algorithm. For a packet movement executed in Line~\ref{move1} of Algorithm~\ref{alg:matroidTransform}, we have $d_{i,e}> \bar \Delta_i^e(B')$, thus, when packet $i_k$ located on $e = e_{i_k}$ is moved to $f_i$, it holds that $\pi_i^{e_{i_k}} = \pi_i^e \geq d_{i,e} > \bar \Delta_i^e(B') =\pi_i^{f_i}$. For packet movements executed in Line~\ref{move2}, then by the choice of player $i\in N_e(B')$ (see Line~\ref{alg:pick}) for the corresponding packet $i_k$ it holds $\pi_i^{e_{i_k}} = \pi_i^e > \bar \Delta^e_i(B') = \pi_i^{f_i}$. In both cases we obtain $ \pi_i^e > \pi_i^{f_i}$.
Hence, every movement of a single packet $i_k$ is in strictly decreasing order of virtual value of the resource.
Note that the virtual cost value $\pi_i^e$ does not depend on the profile $B$. Thus, there are at most $m$ different virtual cost values that a packet $i_k$ of player $i$ can experience, and thus packet $i_k$ can move at most $m-1$ times. The following is an upper bound on the total number of packet movements for all players
$\sum_{i\in N} \rk_i \cdot (m-1) \leq n \cdot m \cdot \rk(\cB).$

It is left to argue that the final output $B'$ has lower cost. We prove this inductively by the different types of packet movements. Consider first a packet movement of type \eqref{eq:viol-delay}. Let $B$ and $B'$ be the profiles before and after packet $i_k$ has been moved from $e$ to $f_i$, respectively. We obtain
\[
\begin{aligned}
C(B')-C(B)&= (c_{f_i}(B')-c_{f_i}(B)+d_{i,f_i}) -(c_e(B)-c_e(B')+d_{i,e})\\
&\leq c_{f_i}(\{i\})+d_{i,f_i} -(c_e(B)-c_e(B')+d_{i,e})\\
&= \bar \Delta^e_i(B) - d_{i,e}+(c_e(B')-c_e(B))\\
&\leq \bar \Delta^e_i(B) - d_{i,e} \quad < \quad 0.
\end{aligned}
\]
The first inequality follows from subadditivity, the second inequality from monotonicity of costs $c_e$. The last strict inequality follows from assumption~\eqref{eq:viol-delay}.

Now consider packet movements of type \eqref{eq:viol-cost}. We treat all movements occurring in one run of the while loop in Line~\ref{alg:second-while}. Let $B$ denote the profile before and $B'$ after all these movements. Let $T_e(B)\subseteq N_e(B)$ denote the set of those players whose packet $i_k$ on $e$ is moved to $f_i$ during the while loop.
Let
$ F_e(B) = \bigcup_{i\in T_e(B)}\{f_i\} $ and for  $i\in T_e(B)$ define
$T_{f_i}(B) =\{j \in T_e(B) \mid f_j=f_i\}.$
We derive some useful observations.
Before entering the while loop, it holds
\begin{equation}\label{eq:cost1} 
c_e(B)>\sum_{i\in N_e(B)}\Big(\bar \Delta_i^e(B)-d_{i,e}\Big)\\
= \sum_{i\in N_e(B)\setminus T_e(B)}\Big(\bar \Delta_i^e(B)-d_{i,e}\Big)+
\sum_{i\in T_e(B)}\Big(\bar \Delta_i^e(B)-d_{i,e}\Big).
\end{equation}
Moreover, after exiting the while loop it holds
\begin{equation}\label{eq:cost2} c_e(B')\leq \sum_{i\in N_e(B)\setminus T_e(B)}\Big(\bar \Delta_i^e(B)-d_{i,e}\Big).\end{equation}
Thus, combining~\eqref{eq:cost1} and~\eqref{eq:cost2} we get 
\begin{equation}\label{eq:cost3}
c_e(B)-c_e(B')>\sum_{i\in T_e(B)}\Big(\bar \Delta_i^e(B)-d_{i,e}\Big).
\end{equation}

Putting everything together, we obtain 
\[
\begin{aligned}
C(B')-C(B)&=\sum_{f_i\in F_e(B)} \left(c_{f_i}(B')-c_{f_i}(B)\right)+\sum_{i\in T_e(B)}d_{i,f_i}-\Big(c_e(B)-c_e(B')+ \sum_{i\in T_e(B)}d_{i,e}\Big)\\
&\leq \sum_{f_i\in F_e(B)}  \sum_{j\in T_{f_i}(B)} c_{f_i}(\{j\})+\sum_{i\in T_e(B)}d_{i,f_i}-\Big(c_e(B)-c_e(B')+ \sum_{i\in T_e(B)}d_{i,e}\Big)\\
&= \sum_{i\in T_e(B)} \bar \Delta_i^e(B)-\Big(c_e(B)-c_e(B')+ \sum_{i\in T_e(B)}d_{i,e}\Big)\quad < \quad 0,
\end{aligned}\]
where the first inequality follows from subadditivity and the last inequality follows from~\eqref{eq:cost3}. 
\end{proof}

\section{Connection Games without Delays}
\label{sec:fixed}

In this section, we study connection games in an undirected graph $G = (V,E)$ with a common source vertex $s \in V$. Every player $i$ wants to connect a player-specific terminal node $t_i \in V$ to $s$. Consequently, every strategy $P_i$ of player $i$ is an $(s,t_i)$-path in $G$. We denote the set of paths for player $i$ by $\cP_i$ and the set of profiles by $\cP$.

Note that when each edge cost contains a player-specific delay component $d_{i,e}$, we can take any multi-source multi-terminal connection game and introduce a new auxiliary source vertex $s$. Then connect $s$ to each $s_i$ with an auxiliary edge $e_i$, which has cost $d_{i,e_i} = 0$ and $d_{j,e_i} = M$, for some prohibitively large constant $M$. Now in any equilibrium and any optimal state of the resulting game, player $i$ will choose an $(s,t_i)$-path which begins with edge $e_i$. Moreover, $e_i$ does not generate additional cost for player $i$. As such, the optimal solutions, the Nash equilibria, and their total costs correspond exactly to the ones of the original multi-source multi-terminal game. Hence, in games with non-shareable player-specific delays, the assumption of a common source is without loss of generality and existing lower bounds on the price of stability apply~\cite{ChenRV10,harks2017}.

In this section, we instead focus on connection games with fixed shareable costs $c_e \ge 0$ and no player-specific delays $d_{i,e} = 0$, for all players $i$ and all edges $e \in E$. For the general multi-terminal multi-source case with such costs, it is straightforward to observe that the greedy algorithm analyzed by Gupta and Kumar~\cite{GuptaK15} can be turned into a separable protocol via the Prim-Sharing idea~\cite{ChenRV10}. This implies that we can obtain separable cost sharing protocols with a constant price of stability in polynomial time.

\begin{proposition}
For every connection game in undirected graphs with fixed costs, there is an enforceable profile that can be computed in polynomial time and yields a separable cost sharing protocol with constant price of stability.
\end{proposition}

For single-source games with fixed costs, existing results for cost sharing games with arbitrary sharing imply that an optimal profile is always enforceable~\cite{AnshelevichDTW08,ChenRV10}. We here provide a significantly stronger result for polynomial-time computation of cheap enforceable profiles.

\begin{theorem}
\label{thm:singleSource}
Let $P$ be a strategy profile for a single-source connection game with fixed costs. There is an enforceable strategy $P'$ with $C(P')\leq C(P)$ that can be computed by Algorithm~\ref{alg:singleSourceTransform} in polynomial time.
\end{theorem}

It is straightforward that for fixed costs we can transform each profile $P$ into a cheaper \emph{tree profile} $\hat{P}$, in which the union of player paths constitute a tree $T$. Over the course of the algorithm, we adjust this tree and construct a cost sharing for it in a bottom-up fashion. The approach has similarities to an approach for obtaining approximate equilibria for single-source cost sharing games with arbitrary sharing~\cite{AnshelevichDTW08}. However, our algorithm exploits crucial properties of separable protocols, thereby providing an exact Nash equilibrium and polynomial running time.

When designing a separable protocol based on a state $\hat{P}$, we can always assume that when a player $i$ deviates unilaterally to one or more edges $e \in G \setminus \hat{P}_i$, she needs to pay all of $c_e$. As such, player $i$ always picks a collection of shortest paths with respect to $c_e$ between pairs of nodes on her current path $\hat{P}_i$. All these paths in $G$ are concisely represented in the algorithm as ``auxiliary edges''. The algorithm initially sets up an auxiliary graph $\hat{G}$ given by $T$ and the set of auxiliary edges based on $\hat{P}$. It adjusts the tree $T$ by removing edges of $T$ and adding auxiliary edges in a structured fashion. 

We first show in the following lemma that this adjustment procedure improves the total cost of the tree, and that the final tree $\hat{T}$ is enforceable in $\hat{G}$. In the corresponding cost sharing, every auxiliary edge contained in $\hat{T}$ is completely paid for by a single player that uses it. In the subsequent proof of the theorem, we only need to show that for the auxiliary edges in $\hat{T}$, the edge costs of the corresponding shortest paths in $G$ can be assigned to the players such that we obtain a Nash equilibrium in $G$. The proof shows that the profile $P'$ evolving in this way is enforceable in $G$ and only cheaper than $P$.

\begin{algorithm}[h!]
\DontPrintSemicolon
\SetKwInput{Input}{Input}
\SetKwInput{Output}{Output}
\Input{Connection game $(N, G, (t_1,\ldots,t_n), s, c)$ and profile $P \in \cP$}
\Output{Enforceable profile $P'$ with $C(P')\leq C(P)$.}
\vspace{0.25cm}

Transform $P$ into a tree profile $\hat{P}$ and let $T \leftarrow \bigcup_i \hat{P}_i$\;
$\hat{c}_e(i) \leftarrow 0$, for all $e \in T, i \in N$\;
Insert $T$ into empty graph $G'$, root $T$ in $s$, number vertices of $T$ in BFS order from $s$\;
\ForEach{$i \in N$ and each $v_k, v_{k'} \in \hat{P}_i$ with $k > k'$}{
    Add to $\hat{G}$ an auxiliary edge $e = (v_k, v_{k'})$\;
    $P(v_k,v_k') \leftarrow$ shortest path in $G$ from $v_k$ to $v_{k'}$\;
    For all $j \in N$, set $\hat{c}_e(j) \leftarrow \sum_{f \in P(v_k,v_{k'})} c_f$\; 
} 
Label every $e \in T$ as ``open''\;

\ForEach{open $e \in T$ in bottom-up order}{
  $\hat{c}_e(i) \leftarrow c_e$, for all $i \in N_e(\hat{P})$\;
  $P_d(i) \leftarrow$ shortest $(s,t_i)$-path in $\hat{G}$ for edge costs $\hat{c}_e(i)$\;
  $\Delta_i^e \leftarrow \sum_{e' \in P_d(i)} \hat{c}_{e'}(i) - \sum_{e' \in \hat{P}_i, e' \neq e} \hat{c}_{e'}(i)$\;
  \If{$c_e \le \sum_{i \in N_e(\hat{P})} \Delta_i^e$}{
    For all $i \in N_e(\hat{P})$, assign $\hat{c}_e(i) \in [0, \Delta_i^e]$ such that $\sum_{i \in N_e(\hat{P})} \hat{c}_e(i) = c_e$\;
    Label $e$ as ``closed''\;
  }
  \Else{
    $D \leftarrow$ set of highest deviation vertices\;
    $P_T(e,v) \leftarrow$ path between higher node of $e$ and $v$ in $T$\;
    Remove all paths $P_T(e,v)$ from $T$ and $\hat{G}$\;
    \ForEach{$v \in D$}{
      Pick one auxiliary edge $e' \in P_d(i)$, for some $i \in N_e(\hat{P})$, such that $e' = (v,u)$ with $u$ a node above $e$ in $T$, and add $e'$ to $T$\;
      $\hat{P}_i \leftarrow P_d(i)$\;
      \ForEach{player $j \neq i$ with $t_j$ below $v$ in $T$}{
         $\hat{P}_j \leftarrow$ ($\hat{P}_j$ from $t_j$ to $v$) $\cup$ ($P_d(i)$ from $v$ to $s$)\;
         $\hat{c}_{e'}(j) \leftarrow 0$\;
      }
      Label $e'$ as ``closed''\;
    }
  }
}
For every $i \in N$, compute $P'_i$ by replacing in $\hat{P}_i$ every auxiliary edge $e=(u,v)$ by the corresponding shortest path $P(u,v)$ in $G$
\caption{Transforming any profile $P$ into an enforceable profile $P'$\label{alg:singleSourceTransform}}
\end{algorithm}

\begin{lemma}
\label{lem:singleSource}
Algorithm~\ref{alg:singleSourceTransform} computes a cost sharing of a feasible tree $\hat{T}$ in the graph $\hat{G}$. The total cost $C(\hat{T}) \le C(T)$, every auxiliary edge in $\hat{T}$ is paid for by a single player, and the corresponding profile $\hat{P}$ is enforceable in $\hat{G}$.
\end{lemma}

\begin{proof}
After building $\hat{G}$, the algorithm considers $T$ rooted in the source $s$. Initially, all edges of $T$ are assumed to have zero cost for all players. All edges of $T$ are labelled ``open''. Our proof works by induction. We assume that players are happy with their strategies $\hat{P}_i$ if all open edges of $T$ have cost 0, all open edges outside $T$ have cost $\hat{c}_e$ for every player, and the closed edges $e \in T$ are shared as determined by $\hat{c}_e$.

The algorithm proceeds in a bottom-up fashion. In an iteration, it restores the cost of an open edge $e$ to its original value. It then considers how much each player $i \in N_e(\hat{P})$ is willing to contribute to $e$. The maximum contribution $\Delta_i^e$ is given by the difference in the cheapest costs to buy an $(s,t_i)$-path for $i$ when (1) $e$ has cost 0 and (2) $e$ has cost $c_e$. By induction, for case (1) we can assume that $i$ is happy with $\hat{P}_i$ when $e$ has cost 0. In case (2), suppose $i$ deviates from (parts of) his current path $\hat{P}_i$ and buys auxiliary edges.

Since by induction $i$ is happy with $\hat{P}_i$ when $e$ has cost 0, there is no incentive to deviate from $\hat{P}_i$ between two vertices of $\hat{P}_i$ below $e$. Moreover, clearly, there is no incentive to deviate from $\hat{P}_i$ between two vertices above $e$ (since edges of $T$ above $e$ are assumed to have zero cost). Hence, if in case (2) the path $P_d(i)$ includes $e$, then $P_d(i) = \hat{P}_i$, so $\Delta_i^e = c_e$. Otherwise, player $i$ finds a path that avoids $e$. By the observations so far, $P_d(i)$ can be assumed to follow $\hat{P}_i$ from $t_i$ up to a vertex $v$, then picks a single auxiliary edge $(v,u)$ to node $u$ above $e$, and then follows $\hat{P}_i$ to $s$. We call the vertex $v$ the \emph{deviation vertex} of $P_d(i)$.

Based on $P_d(i)$, the algorithm computes a maximum contribution $\Delta^e_i$ for each player $i \in N_e(\hat{P})$, which $i$ is willing to pay for edge $e$ currently under consideration. If in total these contributions suffice to pay for $e$, then we determine an arbitrary cost sharing of $c_e$ such that each player $i \in N_e$ pays at most $\Delta_i^e$. Thereby, every player $i \in N_e(\hat{P})$ remains happy with his path $\hat{P}_i$, and the inductive assumptions used above remain true. We can label $e$ as closed and proceed to work on the next open edge in the tree $T$.

Otherwise, if the contributions $\Delta_i^e$ do not suffice to pay for $c_e$, then for every $i \in N_e(\hat{P})$ the path $P_d(i)$ avoids $e$ and contains a deviation vertex. The algorithm needs to drop $e$ and change the strategy of every such player. It considers the ``highest'' subset $D$ of deviation vertices, i.e., the unique subset such that $D$ contains exactly one deviation vertex above each terminal $t_i$. The algorithm removes all edges from $T$ that lie below and including $e$ and above any $v \in D$. For each $v \in D$, it then adds one auxiliary edge from the corresponding $P_d(i)$ to $T$. As observed above, these edges connect $v$ to some node $u$ above $e$, and thereby yield a new feasible tree $T$ in $\hat{G}$. 

Since $P_d(i)$ is a best response for player $i$, we assign $i$ to pay for the cost of the auxiliary edge. After this update, $i$ is clearly happy with $\hat{P}_i$. Moreover, every other player $j \in N_e(\hat{P})$ that now uses the auxiliary edge paid by player $i$ is happy with his new strategy $\hat{P}_j$. The auxiliary edge has cost zero for player $j$, and the path from $t_i$ to the deviation vertex $v$ has not changed. By induction $j$ was happy with this path after we finished paying for the last edge below $v$. Thus, we can label all auxiliary edges added to $T$ as closed and proceed to work on the next open edge in the tree $T$. 

By induction, this proves that the algorithm computes a cost sharing that induces a separable protocol with the final tree $T'$ being a Nash equilibrium in $\hat{G}$. Moreover, if we change the tree during the iteration for edge $e$, it is straightforward to verify that the total cost of the tree strictly decreases. 
\end{proof}

\begin{proof}[Proof of Theorem~\ref{thm:singleSource}]
The previous lemma shows that the algorithm computes a cost sharing of a tree $\hat{T}$ in $\hat{G}$, such that every player is happy with the path $\hat{P}_i$ and every auxiliary edge in $\hat{T}$ is paid for completely by a single player. We now transform $\hat{P}$ into $P'$ by replacing each auxiliary edge $e = (u,v) \in \hat{P}_i$ by the corresponding shortest path $P(u,v)$ in $G$. We denote by $E_i$ the set of edges introduced in the shortest paths for auxiliary edges in $\hat{P}_i$. For the total cost of the resulting profile we have that $C(P') \le C(\hat{P}) \le C(P)$, since the sets $E_i$ can overlap with each other or the non-auxiliary edges of $\hat{T}$.

We show that $P'$ is enforceable by transforming the cost sharing constructed in function $\hat{c}$ into separable cost sharing functions as follows. Initially, set $\xi_{i,e}(P') = 0$ for all $e \in E$ and $i \in N$. Then, for each non-auxiliary edge $e \in \hat{T}$ we assign $\xi_{i,e}(P') = \hat{c}_e(i)$ if $e \in \hat{P}_i$ and $\xi_{i,e}(P') = 0$ otherwise. Finally, number players arbitrarily from 1 to $n$ and proceed in that order. For player $i$, consider the edges in $E_i$. For every $e \in E_i$, if $\sum_{j < i} \xi_{j,e}(P') = 0$, then set $\xi_{i,e}(P') = c_e$. 

This yields a budget-balanced assignment for state $P'$. As usual, if a player $i$ deviates in $P'$ from $P'_i$ to $P''_i$, we can assume player $i$ is assigned to pay the full cost $c_e$ for every edge $e \in P''_i \setminus P'_i$. To show that there is no profitable deviation from $P'$, we first consider a thought experiment, where every edge in $E_i$ comes as a separate edge bought by player $i$. Then, clearly $P'$ is enforceable -- the cost of $P'_i$ with $\xi$ is exactly the same as the cost of $\hat{P}$ with $\hat{c}$ in $\hat{G}$. Moreover, any deviation $P''_i$ can be interpreted as an $(s,t_i)$-path in $\hat{G}$ by replacing all subpaths consisting of non-auxiliary edges in $P''$ by the corresponding auxiliary edge of $\hat{G}$. As such, the cost of $P''_i$ is exactly the same as the cost of the corresponding deviation in $\hat{G}$. Now, there is not a separate copy for every edge in $E_i$. The set $E_i$ can overlap with other sets $E_j$ and/or non-auxiliary edges. Then player $i$ might not need to pay the full cost on some $e \in E_i$. Note, however, every edge for which player $i$ pays less than $c_e$ is present in $P'_i$ as well. Hence, $P''_i$ cannot improve over $P'_i$ due to this property.
\end{proof}

The result continues to hold for various generalizations. For example, we can immediately apply the arguments in directed graphs, where every player $i$ seeks to establish a directed path between $t_i$ and $s$. Moreover, the proof can also be applied readily for a group-connection game, where each player wants to establish a directed path to $s$ from \emph{at least one node of a set $V_i \subset V$}. For this game, we simply add a separate super-terminal $t_i$ for every player $i$ and draw a directed edge of cost 0 from $t_i$ to every node in $V_i$.

\begin{corollary}
Let $P$ be a strategy profile for a single-source group-connection game in directed graphs with fixed costs. There is an enforceable profile $P'$ with $C(P')\leq C(P)$ that can be computed by Algorithm~\ref{alg:singleSourceTransform} in polynomial time.
\end{corollary}

\section{Connection Games and Graph Structure}
\label{sec:nSepa}
%
%
In this section, we consider connection games played in undirected graphs $G=(V,E)$ with player-specific source-terminal pairs. Each player $i \in N$ has a source-terminal-pair $(s_i,t_i)$. Note that we can assume w.l.o.g.\ that $(G,(s_1,t_1),\ldots,(s_n,t_n))$ is \textit{irredundant}, meaning that each edge and each vertex of $G$ is contained in at least one $(s_i,t_i)$-path for some player $i \in N$ (nodes and edges which are not used by any player can easily be recognized (and then deleted) by Algorithm~\nameref{irredundant} at the end of the section; adapted from Algorithm 1 in~\cite{chen2016}). \\
Harks et al.~\cite{harks2017} characterized enforceability for the special case with $d_{i,e}=0$ for all $i \in N, e \in E$ via an LP. We can directly adapt this characterization as follows:
\begin{align*}
\text{LP($P$)} \ \ &\max &\sum_{i \in N, e \in P_i}\xi_{i,e} \\
& \text{s.t.: } & \sum_{i \in N_e(P)} \xi_{i,e} &\leq c_e & \forall e \in E \text{ with } N_e(P)\neq \emptyset\\
& & \sum_{e \in P_i \setminus P_i'} {\left(\xi_{i,e}+d_{i,e}\right)} &\leq \sum_{e \in P_i' \setminus P_i} {\left(c_e+d_{i,e}\right)}  & \forall P_i' \in \mathcal P_i \ \forall i \in N \tag{NE} \label{NE}\\
& &\xi_{i,e} &\geq 0 & \forall e \in P_i \ \forall i \in N 
\end{align*}

\begin{theorem}
The strategy profile $P=(P_1,\ldots,P_n)$ is enforceable if and only if there is an optimal solution $(\xi_{i,e})_{i \in N, e \in P_i}$ for LP($P$) with
\[
\sum_{i\in N_e(P)}{\xi_{i,e}}=c_e \quad \forall e \in E \text{ with } N_e(P)\neq \emptyset. \tag{BB}\label{BB}
\]
\end{theorem}
Given an optimal solution $(\xi_{i,e})_{i \in N, e \in P_i}$ for LP($P$) with the property (\ref{BB}), the profile $P$ becomes a PNE in the game induced by $\xi$, which assigns for each $i \in N$ and $e \in E$ and each strategy profile $P'=(P_1', \ldots, P_n')$ the following cost shares (these cost shares resemble those introduced in \cite{HarksF13}): 
\[
\xi_{i,e}(P') = \begin{cases} \xi_{i,e}, & \text{if } i \in S_e(P')=S_e(P),\\
c_e, & \text{if } i \in (S_e(P')\setminus S_e(P)) \text{ and } i=\min (S_e(P')\setminus S_e(P)), \\
c_e, & \text{if } i \in S_e(P') \subsetneq S_e(P) \text{ and } i=\min S_e(P'), \\
0, & \text{else.}
\end{cases}
\] 

We now introduce a subclass of generalized series-parallel graphs for which we design a polynomial time black-box reduction that computes, for a given strategy profile $P$, an enforceable strategy profile with smaller cost. 
\begin{definition}[$n$-series-parallel graph]
An irredundant graph $(G, (s_1,t_1),\ldots,(s_n,t_n))$ is \textit{$n$-series-parallel} if, for all $i\in N$, the subgraph $G_i$ (induced by $\mathcal P_i$) is created by a sequence of series and/or parallel operations starting from the edge $s_i-t_i$.
 For an edge $e=u-v$, a series operation replaces it by a new vertex $w$ and two edges $u-w, w-v$;
A parallel operation adds to $e=u-v$ a parallel edge $e'=u-v$.

\end{definition}

The following theorem summarizes our results for $n$-series-parallel graphs. 

\begin{theorem}\label{theo_sepa}
If $(G, (s_1,t_1),\ldots,(s_n,t_n))$ is $n$-series-parallel, the following holds:
\begin{enumerate}
	\item[(1)] Given an arbitrary strategy profile $P$, an enforceable strategy profile $P'$ with cost $C(P')\leq C(P)$, and corresponding cost share functions $\xi$, can be computed in polynomial time. 
	\item[(2)] For all cost functions $c,d$, every optimal strategy profile of $(G,(s_1,t_1), \ldots, (s_n,t_n),c,d)$ is enforceable.
	\item[(3)] For all edge costs $c$, an optimal Steiner forest of $(G,(s_1,t_1), \ldots, (s_n,t_n),c)$ can be computed in polynomial time.
\end{enumerate}
\end{theorem}

To prove \autoref{theo_sepa}, we need to introduce some notation. Let $(\xi_{i,e})_{i \in N, e \in P_i}$ be an optimal solution for LP($P$). 
For $i \in N$ and $f \in P_i$, we consider all paths $P_i' \in \mathcal P_i$ with $f \notin P_i'$, $P_i \cup P_i'$ contains a unique cycle $C(P_i')$ and $\sum_{e \in P_i \setminus P_i'}{(\xi_{i,e}+d_{i,e})} = \sum_{e \in P_i' \setminus P_i}{(c_e+d_{i,e})}$. Among all these paths, choose one for which the number of edges in $C(P_i')\cap P_i$ is minimal. The corresponding path $A_{i,f}:=C(P_i')\cap P_i'$ is called a \textit{smallest tight alternative of player $i$ for $f$}. 
If we say that player $i$ substitutes $f$ by using $A_{i,f}$, we mean that the current path $P_i$ of player $i$ is changed by using $A_{i,f}$ instead of the subpath $C(P_i')\cap P_i$ (which contains $f$).
Figure~\ref{alternatives} illustrates the described concepts. 

Note that $A_{i,f}$ is also smallest in the sense that every other (tight) alternative for $f$ substitutes a superset of the edges substituted by $A_{i,f}$.

\begin{figure}[ht] \centering \psset{unit=0.85cm}
 \begin{pspicture}(-0.5,-1)(12.5,1.5) 
\uput[0](-0.8,-0.7){\footnotesize\textbf{(a)} Example for $P_i$ (thick) and all alternative paths with tight inequality in LP($P$).}
\Knoten{0}{0}{s}\nput{180}{s}{$s_i$}
\Knoten{1}{0}{v1}
\Knoten{2}{0}{v2}
\Knoten{3}{0}{v3}
\Knoten{4}{0}{v4}
\Knoten{5}{0}{v5}
\Knoten{6}{0}{v6}
\Knoten{7}{0}{v7}
\Knoten{8}{0}{v8}
\Knoten{9}{0}{v9}
\Knoten{10}{0}{v10}
\Knoten{11}{0}{v11}
\Knoten{4.5}{0}{v12}
\Knoten{10.5}{0}{v13}
\Knoten{3.5}{0.6}{v15}
\Knoten{4.5}{0.8}{v16}
\Knoten{5.5}{0.6}{v17}
\Knoten{9}{0.5}{v18}
\Knoten{12}{0}{t}\nput{0}{t}{$t_i$}
\ncline[linewidth=0.07]{-}{s}{t}
\Kante{v9}{v10}{f}
\ncarc[arcangle=85]{-}{v1}{v3}
\ncarc[arcangle=10]{-}{v3}{v15}
\ncarc[arcangle=10]{-}{v15}{v16}
\ncarc[arcangle=10]{-}{v16}{v17}
\ncarc[arcangle=10]{-}{v17}{v6}
\ncarc[arcangle=80]{-}{v15}{v17}
\ncarc[arcangle=85]{-}{v4}{v5}
\ncarc[arcangle=100]{-}{v7}{t}
\ncarc[arcangle=80]{-}{v8}{t}
\ncarc[arcangle=25]{-}{v8}{v18}
\ncarc[arcangle=25]{-}{v18}{v10}
\ncarc[arcangle=85]{-}{v10}{v11}
\end{pspicture}

  \begin{pspicture}(-0.5,-1)(12.5,1.5) 
\uput[0](-0.8,-0.7){\footnotesize\textbf{(b)} Dashed path $P_i'$ with $f \notin P_i'$, but no unique cycle with $P_i$.}
\Knoten{0}{0}{s}\nput{180}{s}{$s_i$}
\Knoten{1}{0}{v1}
\Knoten{2}{0}{v2}
\Knoten{3}{0}{v3}
\Knoten{4}{0}{v4}
\Knoten{5}{0}{v5}
\Knoten{6}{0}{v6}
\Knoten{7}{0}{v7}
\Knoten{8}{0}{v8}
\Knoten{9}{0}{v9}
\Knoten{10}{0}{v10}
\Knoten{11}{0}{v11}
\Knoten{4.5}{0}{v12}
\Knoten{10.5}{0}{v13}
\Knoten{3.5}{0.6}{v15}
\Knoten{4.5}{0.8}{v16}
\Knoten{5.5}{0.6}{v17}
\Knoten{9}{0.5}{v18}
\Knoten{12}{0}{t}\nput{0}{t}{$t_i$}
\ncline[linestyle=dashed]{-}{s}{v1}
\ncline[linestyle=dashed]{-}{v1}{v2}
\ncline[linestyle=dashed]{-}{v2}{v3}
\ncline[linestyle=dashed]{-}{v6}{v7}
\ncline[linestyle=dashed]{-}{v7}{v8}
\ncline{-}{v3}{v6}
\ncline{-}{v8}{t}
\Kante{v9}{v10}{f}
\ncarc[arcangle=85]{-}{v1}{v3}
\ncarc[arcangle=10,linestyle=dashed]{-}{v3}{v15}
\ncarc[arcangle=10,linestyle=dashed]{-}{v15}{v16}
\ncarc[arcangle=10,linestyle=dashed]{-}{v16}{v17}
\ncarc[arcangle=10,linestyle=dashed]{-}{v17}{v6}
\ncarc[arcangle=80]{-}{v15}{v17}
\ncarc[arcangle=85]{-}{v4}{v5}
\ncarc[arcangle=100]{-}{v7}{t}
\ncarc[arcangle=80,linestyle=dashed]{-}{v8}{t}
\ncarc[arcangle=25]{-}{v8}{v18}
\ncarc[arcangle=25]{-}{v18}{v10}
\ncarc[arcangle=85]{-}{v10}{v11}
\end{pspicture}

 \begin{pspicture}(-0.5,-1)(12.5,1.5) 
\uput[0](-0.8,-0.7){\footnotesize\textbf{(c)} Dashed path $P_i'$ with $f \notin P_i'$, unique cycle $C(P_i')$ with $P_i$, but not smallest for $f$.}
\Knoten{0}{0}{s}\nput{180}{s}{$s_i$}
\Knoten{1}{0}{v1}
\Knoten{2}{0}{v2}
\Knoten{3}{0}{v3}
\Knoten{4}{0}{v4}
\Knoten{5}{0}{v5}
\Knoten{6}{0}{v6}
\Knoten{7}{0}{v7}
\Knoten{8}{0}{v8}
\Knoten{9}{0}{v9}
\Knoten{10}{0}{v10}
\Knoten{11}{0}{v11}
\Knoten{4.5}{0}{v12}
\Knoten{10.5}{0}{v13}
\Knoten{3.5}{0.6}{v15}
\Knoten{4.5}{0.8}{v16}
\Knoten{5.5}{0.6}{v17}
\Knoten{9}{0.5}{v18}
\Knoten{12}{0}{t}\nput{0}{t}{$t_i$}
\ncline[linestyle=dashed]{-}{s}{v1}
\ncline[linestyle=dashed]{-}{v1}{v2}
\ncline[linestyle=dashed]{-}{v2}{v3}
\ncline[linestyle=dashed]{-}{v6}{v7}
\ncline[linestyle=dashed]{-}{v7}{v8}
\ncline[linestyle=dashed]{-}{v3}{v6}
\ncline{-}{v8}{t}
\Kante{v9}{v10}{f}
\ncarc[arcangle=85]{-}{v1}{v3}
\ncarc[arcangle=10]{-}{v3}{v15}
\ncarc[arcangle=10]{-}{v15}{v16}
\ncarc[arcangle=10]{-}{v16}{v17}
\ncarc[arcangle=10]{-}{v17}{v6}
\ncarc[arcangle=80]{-}{v15}{v17}
\ncarc[arcangle=85]{-}{v4}{v5}
\ncarc[arcangle=100]{-}{v7}{t}
\ncarc[arcangle=80,linestyle=dashed]{-}{v8}{t}
\ncarc[arcangle=25]{-}{v8}{v18}
\ncarc[arcangle=25]{-}{v18}{v10}
\ncarc[arcangle=85]{-}{v10}{v11}
\end{pspicture}

 \begin{pspicture}(-0.5,-1)(12.5,1.5) 
\uput[0](-0.8,-0.7){\footnotesize\textbf{(d)} Substituting $f$ by using smallest tight alternative $A_{i,f}$.}
\Knoten{0}{0}{s}\nput{180}{s}{$s_i$}
\Knoten{1}{0}{v1}
\Knoten{2}{0}{v2}
\Knoten{3}{0}{v3}
\Knoten{4}{0}{v4}
\Knoten{5}{0}{v5}
\Knoten{6}{0}{v6}
\Knoten{7}{0}{v7}
\Knoten{8}{0}{v8}
\Knoten{9}{0}{v9}
\Knoten{10}{0}{v10}
\Knoten{11}{0}{v11}
\Knoten{4.5}{0}{v12}
\Knoten{10.5}{0}{v13}
\Knoten{3.5}{0.6}{v15}
\Knoten{4.5}{0.8}{v16}
\Knoten{5.5}{0.6}{v17}
\Knoten{9}{0.5}{v18}
\Knoten{12}{0}{t}\nput{0}{t}{$t_i$}
\ncline[linewidth=0.07]{-}{s}{v1}
\ncline[linewidth=0.07]{-}{v1}{v2}
\ncline[linewidth=0.07]{-}{v2}{v3}
\ncline[linewidth=0.07]{-}{v6}{v7}
\ncline[linewidth=0.07]{-}{v7}{v8}
\ncline[linewidth=0.07]{-}{v3}{v6}
\ncline{-}{v8}{v10}
\ncline[linewidth=0.07]{-}{v10}{t}
\Kante{v9}{v10}{f}
\ncarc[arcangle=85]{-}{v1}{v3}
\ncarc[arcangle=10]{-}{v3}{v15}
\ncarc[arcangle=10]{-}{v15}{v16}
\ncarc[arcangle=10]{-}{v16}{v17}
\ncarc[arcangle=10]{-}{v17}{v6}
\ncarc[arcangle=80]{-}{v15}{v17}
\ncarc[arcangle=85]{-}{v4}{v5}
\ncarc[arcangle=100]{-}{v7}{t}
\ncarc[arcangle=80]{-}{v8}{t}\nbput[npos=0.45]{$A_{i,f}$}
\ncarc[arcangle=25,linewidth=0.07]{-}{v8}{v18}
\ncarc[arcangle=25,linewidth=0.07]{-}{v18}{v10}
\ncarc[arcangle=85]{-}{v10}{v11}
\end{pspicture}
\caption{Illustration of the introduced concepts.}\label{alternatives}
\end{figure}
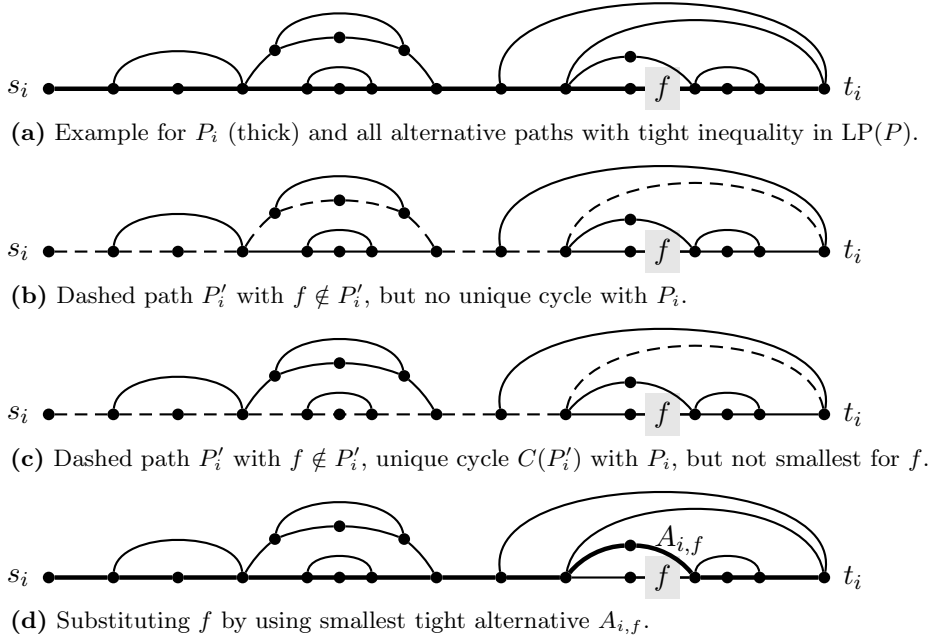

We are now able to prove \autoref{theo_sepa}.

\begin{proof}[Proof of \autoref{theo_sepa}]
We first describe how to compute, given an arbitrary strategy profile $P=(P_1,\ldots,P_n)$, an enforceable strategy profile with cost at most $C(P)$. 

Assume that $P$ is not enforceable (otherwise there is nothing to do). 
Let $(\xi_{i,e})_{i \in N, e \in P_i}$ be an optimal solution for LP($P$). 
In the following, we denote the variables  $(\xi_{i,e})_{i \in N, e \in P_i}$ as \textit{cost shares}, although they do not correspond to a budget-balanced cost sharing protocol (since $P$ is not enforceable).
There is at least one edge $f$ which is not completely paid, i.e. for which $\sum_{i \in N_f(P)} \xi_{i,f} < c_f$ holds. The optimality of the cost shares $(\xi_{i,e})_{i \in N, e \in P_i}$ for LP($P$) implies that each player $i \in N_f(P)$, i.e. each user of $f$, has an alternative path $P_i'$ with $f \notin P_i'$, for which equality holds in the corresponding LP($P$)-inequality (otherwise increasing $\xi_{i,f}$ by a small amount, while all other cost shares remain unchanged, would yield a feasible LP-solution with higher objective function value). Using the notation introduced above, each user $i$ of $f$ has a smallest tight alternative $A_{i,f}$ for $f$. 
Furthermore, if $P_i$ contains more than one edge which is not completely paid, there is a combination of smallest tight alternatives so that all edges which are not completely paid are substituted (see Figure~\ref{combination}, where $f,g,h$ are not completely paid and we substitute all these edges by combining $A_{i,g}$ and $A_{i,h}$). 

\begin{figure}[ht] \centering \psset{unit=0.8cm}
 \begin{pspicture}(-0.5,-0.2)(12.5,1.5) 
\Knoten{0}{0}{s}\nput{180}{s}{$s_i$}
\Knoten{1}{0}{v1}
\Knoten{2}{0}{v2}
\Knoten{3}{0}{v3}
\Knoten{4}{0}{v4}
\Knoten{5}{0}{v5}
\Knoten{6}{0}{v6}
\Knoten{7}{0}{v7}
\Knoten{8}{0}{v8}
\Knoten{9}{0}{v9}
\Knoten{10}{0}{v10}
\Knoten{11}{0}{v11}
\Knoten{4.5}{0}{v12}
\Knoten{10.5}{0}{v13}
\Knoten{3.5}{0.6}{v15}
\Knoten{4.5}{0.8}{v16}
\Knoten{5.5}{0.6}{v17}
\Knoten{9}{0.5}{v18}
\Knoten{12}{0}{t}\nput{0}{t}{$t_i$}
\ncline[linewidth=0.07]{-}{s}{v1}
\ncline{-}{v2}{v3}
\ncline[linewidth=0.07]{-}{v6}{v7}
\ncline[linewidth=0.07]{-}{v7}{v8}
\ncline[linewidth=0.07]{-}{v3}{v6}
\ncline{-}{v8}{v10}
\ncline{-}{v10}{t}
\Kante{v9}{v10}{f}
\Kante{v1}{v2}{g}
\Kante{v11}{t}{h}
\ncarc[arcangle=85,linewidth=0.07]{-}{v1}{v3}
\ncarc[arcangle=10]{-}{v3}{v15}
\ncarc[arcangle=10]{-}{v15}{v16}
\ncarc[arcangle=10]{-}{v16}{v17}
\ncarc[arcangle=10]{-}{v17}{v6}
\ncarc[arcangle=80]{-}{v15}{v17}
\ncarc[arcangle=85]{-}{v4}{v5}
\ncarc[arcangle=100]{-}{v7}{t}
\ncarc[arcangle=80,linewidth=0.07]{-}{v8}{t}
\ncarc[arcangle=25]{-}{v8}{v18}
\ncarc[arcangle=25]{-}{v18}{v10}
\ncarc[arcangle=85]{-}{v10}{v11}
\end{pspicture}
\caption{Illustration of combining smallest tight alternatives.}\label{combination}
\end{figure}
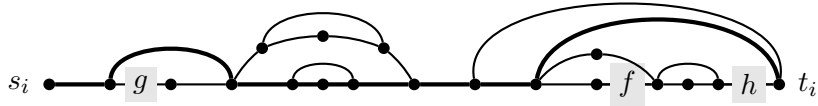

We now consider the strategy profile $P'=(P_1',\ldots,P_n')$ which results from $P$ if all players with unpaid edges in their paths substitute all these edges by a combination of smallest tight alternatives.
Furthermore we define cost shares (again not necessarily budget-balanced) for $P'$ as follows: For each player $i$ and each edge $e \in P_i'$: 
\[
\xi_{i,e}(P')=\begin{cases} \xi_{i,e}, & \text{for } e \in P_i' \cap P_i, \\
c_e, & \text{for } e \in P_i' \setminus P_i.
\end{cases}
\]
Note that the private cost of player $i$ under $P$ equals the private cost of $i$ under $P'$ since the players use tight alternatives. Furthermore note that $\sum_{i \in N_e(P')} \xi_{i,e}(P') > c_e$ is possible (for example if there are two players which did not use an edge $e$ with $c_e>0$ in their paths under $P$, but use it in $P'$ and therefore both pay $c_e$). 
If $\sum_{i \in N_e(P')} \xi_{i,e}(P') \geq c_e$ holds for all edges $e$ with $N_e(P')\neq \emptyset$, we found a strategy profile with the desired properties: $P'$ is cheaper than $P$ since 
\begin{align*}
C(P) &=\sum_{e \in E: N_e(P)\neq \emptyset}{c_e}+\sum_{i \in N}{\sum_{e \in P_i}{d_{i,e}}} \\
&>\sum_{i \in N}\sum_{e \in P_i}{\left(\xi_{i,e}+d_{i,e}\right)} \\
&=\sum_{i \in N}\sum_{e \in P_i'}{\left(\xi_{i,e}(P')+d_{i,e}\right)} \\
&=\sum_{e \in E:  N_e(P')\neq \emptyset}{\sum_{i \in N_e(P')}{\xi_{i,e}(P')}}+\sum_{i \in N}{\sum_{e \in P_i'}{d_{i,e}}} \\ 
&\geq \sum_{e \in E:  N_e(P')\neq \emptyset}{c_e}+\sum_{i \in N}{\sum_{e \in P_i'}{d_{i,e}}}=C(P').
\end{align*}
The strict inequality holds since $P$ is not enforceable, the following equality because the private costs remain unchanged, and the last inequality because of our assumption above. 
Furthermore, $P'$ is enforceable since the cost shares $(\xi_{i,e}(P'))_{i \in N, e \in P_i'}$ induce a feasible solution of LP($P'$) with (\ref{BB}) if we decrease the cost shares for overpaid edges arbitrarily until we reach budget-balance.

Thus assume that there is at least one edge $f$ for which $\sum_{i \in N_f(P')} \xi_{i,f}(P') < c_f$ holds.
First note, for each player $i\in N_f(P')$, that $f \in P_i$ has to hold, since all edges in $P_i' \setminus P_i$ are completely paid (player $i$ pays $c_e$ for $e \in P_i' \setminus P_i$). 
As we will show below, all $i \in N_f(P')$ have a smallest tight alternative $A_{i,f}$ for $f$. We can therefore again update the strategy profile (resulting in $P''$) by letting all players deviate from all nonpaid edges using a combination of smallest tight alternatives. Figure~\ref{2ndstep} illustrates this second phase of deviation, where the edges $r$ and $s$ are now not completely paid. Note that $A_{i,r}$ is not unique in this example, and to use $A_{i,s}$, we need to deviate from $A_{i,h}$ (which player $i$ uses in $P_i'$). 

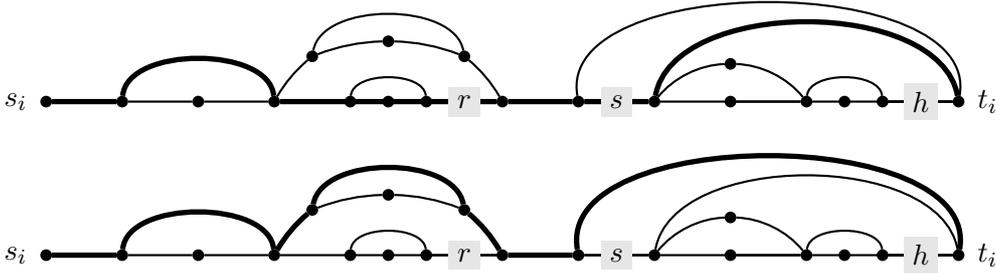
\begin{figure}[ht] \centering \psset{unit=1cm}
 \begin{pspicture}(-0.5,-0.5)(12.5,1.5) 
\Knoten{0}{0}{s}\nput{180}{s}{$s_i$}
\Knoten{1}{0}{v1}
\Knoten{2}{0}{v2}
\Knoten{3}{0}{v3}
\Knoten{4}{0}{v4}
\Knoten{5}{0}{v5}
\Knoten{6}{0}{v6}
\Knoten{7}{0}{v7}
\Knoten{8}{0}{v8}
\Knoten{9}{0}{v9}
\Knoten{10}{0}{v10}
\Knoten{11}{0}{v11}
\Knoten{4.5}{0}{v12}
\Knoten{10.5}{0}{v13}
\Knoten{3.5}{0.6}{v15}
\Knoten{4.5}{0.8}{v16}
\Knoten{5.5}{0.6}{v17}
\Knoten{9}{0.5}{v18}
\Knoten{12}{0}{t}\nput{0}{t}{$t_i$}
\ncline[linewidth=0.07]{-}{s}{v1}
\ncline{-}{v2}{v3}
\ncline[linewidth=0.07]{-}{v6}{v7}
\ncline[linewidth=0.07]{-}{v7}{v8}
\ncline[linewidth=0.07]{-}{v3}{v6}
\ncline{-}{v8}{v10}
\ncline{-}{v10}{t}
\ncline{v9}{v10}
\ncline{v1}{v2}
\ncline{v11}{t}
\Kante{v5}{v6}{r}
\Kante{v7}{v8}{s}
\Kante{v11}{t}{h}
\ncarc[arcangle=85,linewidth=0.07]{-}{v1}{v3}
\ncarc[arcangle=10]{-}{v3}{v15}
\ncarc[arcangle=10]{-}{v15}{v16}
\ncarc[arcangle=10]{-}{v16}{v17}
\ncarc[arcangle=10]{-}{v17}{v6}
\ncarc[arcangle=80]{-}{v15}{v17}
\ncarc[arcangle=85]{-}{v4}{v5}
\ncarc[arcangle=100]{-}{v7}{t}
\ncarc[arcangle=80,linewidth=0.07]{-}{v8}{t}
\ncarc[arcangle=25]{-}{v8}{v18}
\ncarc[arcangle=25]{-}{v18}{v10}
\ncarc[arcangle=85]{-}{v10}{v11}
\end{pspicture}

\begin{pspicture}(-0.5,-0.5)(12.5,1.5) 
\Knoten{0}{0}{s}\nput{180}{s}{$s_i$}
\Knoten{1}{0}{v1}
\Knoten{2}{0}{v2}
\Knoten{3}{0}{v3}
\Knoten{4}{0}{v4}
\Knoten{5}{0}{v5}
\Knoten{6}{0}{v6}
\Knoten{7}{0}{v7}
\Knoten{8}{0}{v8}
\Knoten{9}{0}{v9}
\Knoten{10}{0}{v10}
\Knoten{11}{0}{v11}
\Knoten{4.5}{0}{v12}
\Knoten{10.5}{0}{v13}
\Knoten{3.5}{0.6}{v15}
\Knoten{4.5}{0.8}{v16}
\Knoten{5.5}{0.6}{v17}
\Knoten{9}{0.5}{v18}
\Knoten{12}{0}{t}\nput{0}{t}{$t_i$}
\ncline[linewidth=0.07]{-}{s}{v1}
\ncline{-}{v2}{v3}
\ncline[linewidth=0.07]{-}{v6}{v7}
\ncline{-}{v7}{v8}
\ncline{-}{v3}{v6}
\ncline{-}{v8}{v10}
\ncline{-}{v10}{t}
\ncline{v9}{v10}
\ncline{v1}{v2}
\ncline{v11}{t}
\Kante{v5}{v6}{r}
\Kante{v7}{v8}{s}
\Kante{v11}{t}{h}
\ncarc[arcangle=85,linewidth=0.07]{-}{v1}{v3}
\ncarc[arcangle=10,linewidth=0.07]{-}{v3}{v15}
\ncarc[arcangle=10]{-}{v15}{v16}
\ncarc[arcangle=10]{-}{v16}{v17}
\ncarc[arcangle=10,linewidth=0.07]{-}{v17}{v6}
\ncarc[arcangle=80,linewidth=0.07]{-}{v15}{v17}
\ncarc[arcangle=85]{-}{v4}{v5}
\ncarc[arcangle=100,linewidth=0.07]{-}{v7}{t}
\ncarc[arcangle=80]{-}{v8}{t}
\ncarc[arcangle=25]{-}{v8}{v18}
\ncarc[arcangle=25]{-}{v18}{v10}
\ncarc[arcangle=85]{-}{v10}{v11}
\end{pspicture}
\caption{Illustration for the second phase of deviation.}\label{2ndstep}
\end{figure}
The cost shares are again adapted, that means for each player $i$ and each edge $e \in P_i''$:
\[
\xi_{i,e}(P'')=\begin{cases} \xi_{i,e}, & \text{for } e \in P_i'' \cap P_i, \\
c_e, & \text{for } e \in P_i'' \setminus P_i.
\end{cases}
\]
It is clear that the private costs of the players again remain unchanged and therefore, if all edges are now completely paid, the cost of $P''$ is smaller than the cost of $P$ and $P''$ is enforceable. 

We now show that the tight alternatives used in the second phase of deviation exist. Assume, by contradiction, that there is a player $j$, an edge $f \in P_j'$ which is not completely paid according to $P'$, and player $j$ has no tight alternative for $f$.
Now recall that, whenever an edge $f$ is not completely paid in $P'$, all users $i \in N_f(P')$ already used $f$ in $P$ and therefore $\xi_{i,f}(P')=\xi_{i,f}$ holds for all $i \in N_f(P')\subseteq N_f(P)$. 
Furthermore $f$ was completely paid according to the cost shares of $P$ since we substituted all unpaid edges in the first phase of deviation from $P$ to $P'$. We get 
\[
\sum_{i \in N_f(P')}{\xi_{i,f}}+\sum_{i \in N_f(P)\setminus N_f(P')}{\xi_{i,f}}=c_f>\sum_{i \in N_f(P')}{\xi_{i,f}(P')}=\sum_{i \in N_f(P')}{\xi_{i,f}},
\]
thus there has to be at least one player $k$ which used $f$ in $P_k$, but not in $P_k'$, and with $\xi_{k,f}>0$. Let $A_{k,g}$ be the smallest tight alternative that player $k$ used (to substitute the edge $g$ which was not completely paid in $P$), and also substituted $f$. The situation is illustrated in Figure~\ref{fig_altexist}. 
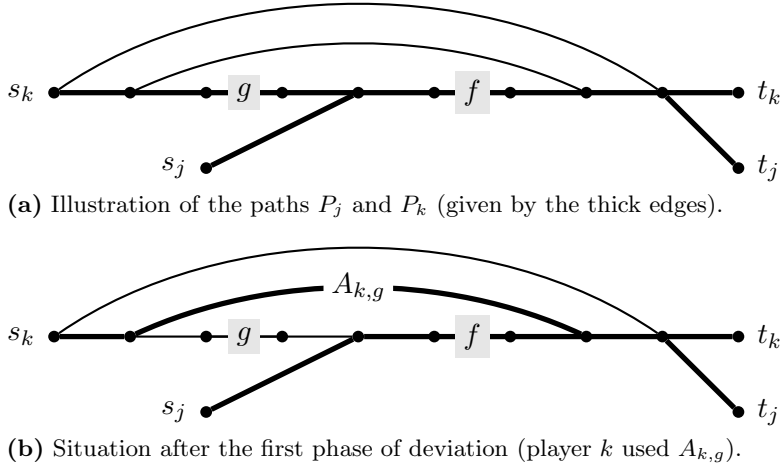
\begin{figure}[ht] \centering \psset{unit=1cm}
 \begin{pspicture}(-0.5,-1.7)(9.5,1.5) 
\uput[0](-0.8,-1.5){\footnotesize\textbf{(a)} Illustration of the paths $P_j$ and $P_k$ (given by the thick edges).}
\Knoten{0}{0}{s}\nput{180}{s}{$s_k$}
\Knoten{1}{0}{v1}
\Knoten{2}{0}{v2}
\Knoten{3}{0}{v3}
\Knoten{4}{0}{v4}
\Knoten{5}{0}{v5}
\Knoten{6}{0}{v6}
\Knoten{7}{0}{v7}
\Knoten{8}{0}{v8}
\Knoten{9}{0}{t}\nput{0}{t}{$t_k$}
\Knoten{9}{-1}{t2}\nput{0}{t2}{$t_j$}
\Knoten{2}{-1}{s2}\nput{180}{s2}{$s_j$}
\ncline[linewidth=0.07]{s}{t}
\Kante{v2}{v3}{g}
\Kante{v5}{v6}{f}
\ncarc[arcangle=35]{-}{s}{v8}
\ncarc[arcangle=25]{-}{v1}{v7}
\ncline[linewidth=0.07]{s2}{v4}
\ncline[linewidth=0.07]{v8}{t2}
\end{pspicture}

 \begin{pspicture}(-0.5,-1.5)(9.5,1.5) 
\uput[0](-0.8,-1.5){\footnotesize\textbf{(b)} Situation after the first phase of deviation (player $k$ used $A_{k,g}$).}
\Knoten{0}{0}{s}\nput{180}{s}{$s_k$}
\Knoten{1}{0}{v1}
\Knoten{2}{0}{v2}
\Knoten{3}{0}{v3}
\Knoten{4}{0}{v4}
\Knoten{5}{0}{v5}
\Knoten{6}{0}{v6}
\Knoten{7}{0}{v7}
\Knoten{8}{0}{v8}
\Knoten{9}{0}{t}\nput{0}{t}{$t_k$}
\Knoten{9}{-1}{t2}\nput{0}{t2}{$t_j$}
\Knoten{2}{-1}{s2}\nput{180}{s2}{$s_j$}
\ncline[linewidth=0.07]{s}{v1}
\ncline{v1}{v4}
\ncline[linewidth=0.07]{v4}{t}
\Kante{v2}{v3}{g}
\Kante{v5}{v6}{f}
\ncarc[arcangle=35]{-}{s}{v8}
\ncarc[arcangle=25,linewidth=0.07]{-}{v1}{v7}\ncput*{$A_{k,g}$}
\ncline[linewidth=0.07]{s2}{v4}
\ncline[linewidth=0.07]{v8}{t2}
\end{pspicture}
\caption{Illustration for the proof that tight alternatives exist.}\label{fig_altexist}
\end{figure}
We now show that the LP($P$)-solution cannot be optimal. Since player $j$ has no tight alternative for $f$, increasing $\xi_{j,f}$ by some suitably small amount, and decreasing $\xi_{k,f}$ by the same amount, yields a feasible LP($P$)-solution. But now player $k$ has no tight alternative for $g$ anymore, since all tight alternatives for $g$ also substitute $f$. Therefore we can increase $\xi_{k,g}$ by some small amount, leading to a feasible LP($P$)-solution with higher objective function value, contradiction.
Thus we showed that the tight alternatives used in the second phase of deviation exist. 

As already mentioned above, if all edges in $P''$ are completely paid, $P''$ is enforceable and cheaper than $P$ and we are finished. Thus we again assume that there is at least one edge which is not completely paid. Analogously as for $P'$ we can show that, for each such edge $f$ and each player $i \in N_f(P'')$, $f \in P_i$ has to hold. 
Furthermore all users of a nonpaid edge have a tight alternative for this edge (the proof that this holds is a little bit more complicated as above, we possibly need to involve three players now): 
Assume that a player $i$ does not have a tight alternative for an edge $f \in P_i$ which is not completely paid according to $P''$, but was completely paid before (i.e. according to $P'$ and also according to $P$). 
Thus there has to be a player $j\in N_f(P)$ with $\xi_{j,f}>0$ who deviated from $f$ by using $A_{j,g}$ in some phase before. If player $j$ did this in the first phase of deviation, the edge $g$ was not completely paid according to $P$ and we can change the LP($P$)-solution as described above to get a contradiction. If the deviation happened in the second phase, the edge $g$ was completely paid according to $P$. Thus there has to be a third player $k$ (but $k=i$ possible) with $\xi_{k,g}>0$ which used some $A_{k,h}$ in the first phase of deviation that also substituted $g$. We are now able to change the cost shares of $P$ to get a contradiction: First, player $i$ increases $\xi_{i,f}$, while player $j$ decreases $\xi_{j,f}$. Now player $j$ increases $\xi_{j,g}$, while player $k$ decreases $\xi_{k,g}$. Finally player $k$ increases $\xi_{k,h}$. By suitably small changes, we get a feasible solution for LP($P$) with higher objective function value than the original optimal cost shares; contradiction. 
Therefore, in a third phase of deviation, all players with nonpaid edges deviate from all those edges by a combination of smallest tight alternatives.

If we proceed in this manner, we finally have to reach a strategy profile for which all edges are completely paid (and thus it is enforceable and cheaper than the profile $P$): In each phase of deviation, at least one edge is substituted by all players which use this edge in $P$. Furthermore, players never return to substituted edges. Therefore, after at most $|P|$ phases of deviation, we reach a strategy profile with the desired property (where $|P|$ denotes the number of edges in the union of the paths $P_1,\ldots,P_n$). 
The existence of the needed tight alternatives in the $k$th phase of deviation can be shown as follows: Assume that $P^{(k)}$ is the current strategy profile, $f$ an edge which is not completely paid according to $P^{(k)}$, and there is a player $i$ which uses an edge $f$, but has no tight alternative for it. 
Then there has to be a player $j$ with $\xi_{j,f}>0$ who deviated from $f$ in some phase $\ell \leq k-1$ by using $A_{j,g}$, where $g$ was not completely paid in the corresponding strategy profile $P^{(\ell)}$. If $\ell=1$ holds, we can decrease $\xi_{j,f}$ and increase $\xi_{i,f},\xi_{j,g}$; contradiction. For $\ell \geq 2$, the edge $g$ was completely paid in $P$ and therefore, there has to be a player $p$ with $\xi_{p,g}>0$ which substituted $g$ in some phase $\leq \ell-1$ by using $A_{p,h}$, and so on. This yields a sequence of players and edges $(i,f), (j,g), (p,h), \ldots, (q,s)$, where the edge $s$  was not completely paid according to $P$. We can now change the cost shares along this sequence (as described above for the third phase of deviation) to get a contradiction. 

Algorithm~\nameref{alg} summarizes the described procedure for computing an enforceable strategy profile $P'$ with cost $C(P')\leq C(P)$ and corresponding cost share functions $\xi$. To complete the proof of the first statement of \autoref{theo_sepa}, it remains to show that $P'$ and $\xi$ can be computed in polynomial time, i.e. Algorithm~\nameref{alg} has polynomial running time. 
As a first step, we show how to compute an optimal solution for LP($P$) in polynomial time. 
To this end we show that, for every player $i$, we do not need to consider all paths $P_i' \in \mathcal P_i$ in (\ref{NE}) of LP($P$), which can be exponentially many paths, but only a set of \textit{alternatives} $\mathcal A_i$ of polynomial cardinality. 
Recall that the graph $G_i$ (induced by $\mathcal P_i$) essentially looks like displayed in Figure~\ref{structureG_i}, and we can w.l.o.g. assume that $P_i$ is given by the thick edges. 

\begin{figure}[ht] \centering \psset{unit=1cm}
 \begin{pspicture}(-0.5,-0.5)(12.5,1.5) 
\Knoten{0}{0}{s}\nput{180}{s}{$s_i$}
\Knoten{1}{0}{v1}
\Knoten{2}{0}{v2}
\Knoten{3}{0}{v3}
\Knoten{4}{0}{v4}
\Knoten{5}{0}{v5}
\Knoten{6}{0}{v6}
\Knoten{7}{0}{v7}
\Knoten{8}{0}{v8}
\Knoten{9}{0}{v9}
\Knoten{10}{0}{v10}
\Knoten{11}{0}{v11}
\Knoten{4.5}{0}{v12}
\Knoten{10.5}{0}{v13}
\Knoten{3.5}{0.6}{v15}
\Knoten{4.5}{0.8}{v16}
\Knoten{5.5}{0.6}{v17}
\Knoten{9}{0.5}{v18}
\Knoten{2}{0.75}{v19}
\Knoten{12}{0}{t}\nput{0}{t}{$t_i$}
\ncline[linewidth=0.07]{-}{s}{t}
\ncline{v9}{v10}
\ncarc[arcangle=50]{-}{v1}{v3}
\ncarc[arcangle=35]{-}{v1}{v19}
\ncarc[arcangle=35]{-}{v19}{v3}
\ncarc[arcangle=10]{-}{v3}{v15}
\ncarc[arcangle=10]{-}{v15}{v16}
\ncarc[arcangle=10]{-}{v16}{v17}
\ncarc[arcangle=10]{-}{v17}{v6}
\ncarc[arcangle=80]{-}{v15}{v17}
\ncarc[arcangle=85]{-}{v4}{v5}
\ncarc[arcangle=100]{-}{v7}{t}
\ncarc[arcangle=80]{-}{v8}{t}
\ncarc[arcangle=25]{-}{v8}{v18}
\ncarc[arcangle=25]{-}{v18}{v10}
\ncarc[arcangle=85]{-}{v10}{v11}
\end{pspicture}
\caption{Structure of $G_i$, where $P_i$ is given by the thick edges.}\label{structureG_i}
\end{figure}
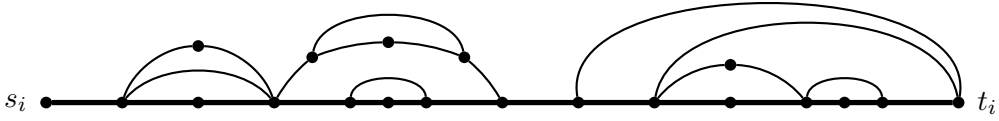

An arbitrary $(s_i,t_i)$-path $P_i'$ consists of subpaths of $P_i$ together with some of the ``arcs''. We call these arcs \textit{alternatives (according to $P_i$)}, and formally, an alternative is a path $A$ which connects two nodes of $P_i$, but is edge-disjoint with $P_i$. The subpath of $P_i$ with the same endnodes as $A$ is denoted by $P_i^A$, and we say that this subpath is \textit{substituted by $A$} (cf. Figure~\ref{fig_alt} for illustration). 
Note that there can be different alternatives which substitute the same subpath of $P_i$ (in Figure~\ref{fig_alt}, this holds for example for the two arcs on the left which both substitute the second and third edge of $P_i$). 
Whenever this is the case, we choose such an alternative with smallest sum of edge costs plus player $i$'s delays, and denote this alternative $A$ as a \textit{cheapest} alternative for $P_i^A$. 
Let $\mathcal A_i$ be the set of all cheapest alternatives according to $P_i$. It is clear that 
\[
\sum_{e \in P_i \setminus P_i'}{\left(\xi_{i,e}+d_{i,e}\right)} \leq \sum_{e \in P_i' \setminus P_i}{\left(c_e+d_{i,e}\right)} \quad \forall P_i' \in \mathcal P_i
\]
holds if and only if
\[
\sum_{e \in P_i^{A}}{\left(\xi_{i,e}+d_{i,e}\right)} \leq \sum_{e \in A}{\left(c_e+d_{i,e}\right)} \quad \forall A \in \mathcal A_i
\]
holds. Since the paths in $\mathcal A_i$ are edge-disjoint, $|\mathcal A_i|$ is bounded by $|E|$. Algorithm~\nameref{algpaths} computes $\mathcal A_i$ in polynomial time. Thus, we can solve LP($P$) in polynomial time. 

\begin{figure}[ht] \centering \psset{unit=1cm}
 \begin{pspicture}(-0.5,-0.5)(12.5,1.5) 
\Knoten{0}{0}{s}\nput{180}{s}{$s_i$}
\Knoten{1}{0}{v1}
\Knoten{2}{0}{v2}
\Knoten{3}{0}{v3}
\Knoten{4}{0}{v4}
\Knoten{5}{0}{v5}
\Knoten{6}{0}{v6}
\Knoten{7}{0}{v7}
\Knoten{8}{0}{v8}
\Knoten{9}{0}{v9}
\Knoten{10}{0}{v10}
\Knoten{11}{0}{v11}
\Knoten{4.5}{0}{v12}
\Knoten{10.5}{0}{v13}
\Knoten{3.5}{0.6}{v15}
\Knoten{4.5}{0.8}{v16}
\Knoten{5.5}{0.6}{v17}
\Knoten{9}{0.5}{v18}
\Knoten{2}{0.75}{v19}
\Knoten{12}{0}{t}\nput{0}{t}{$t_i$}
\ncline{-}{s}{v7}
\ncline[linewidth=0.07]{-}{v7}{t}
\ncline{v9}{v10}
\ncarc[arcangle=50]{-}{v1}{v3}
\ncarc[arcangle=35]{-}{v1}{v19}
\ncarc[arcangle=35]{-}{v19}{v3}
\ncarc[arcangle=10]{-}{v3}{v15}
\ncarc[arcangle=10]{-}{v15}{v16}
\ncarc[arcangle=10]{-}{v16}{v17}
\ncarc[arcangle=10]{-}{v17}{v6}
\ncarc[arcangle=80]{-}{v15}{v17}
\ncarc[arcangle=85]{-}{v4}{v5}
\ncarc[linestyle=dashed, arcangle=100]{-}{v7}{t}\ncput{\colorbox{almostwhite}{$A$}}
\ncarc[arcangle=80]{-}{v8}{t}
\ncarc[arcangle=25]{-}{v8}{v18}
\ncarc[arcangle=25]{-}{v18}{v10}
\ncarc[arcangle=85]{-}{v10}{v11}
\end{pspicture}
\caption{Alternative $A$ (dashed); substituted subpath $P_i^A$ (thick).}\label{fig_alt}
\end{figure}
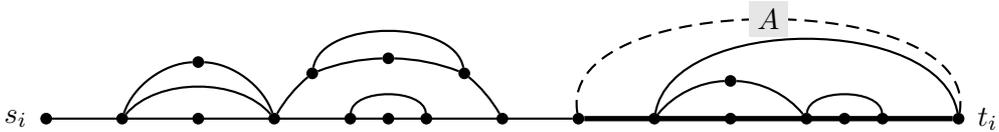

To complete the proof that Algorithm~\nameref{alg} has polynomial running time, it remains to show that the combination of smallest tight alternatives in Line 7 of Algorithm~\nameref{alg} can be found in polynomial time (recall that there will be at most $|P|\leq |E|$ calls of the repeat-loop; and all other steps are obviously polynomial). 
Since $|\mathcal A_i|\leq |E|$ holds for all $i \in N$, we can find, for each edge $f$ which is not completely paid and each user $i$ of $f$, a smallest tight alternative $A_{i,f}$ in polynomial time. If player $i$ uses more than one edge which is not completely paid, a combination of the corresponding smallest tight alternatives can also easily be found; thus step 7 is polynomial.

Overall we showed that the first statement of \autoref{theo_sepa} holds. 
The second statement, i.e. that every optimal strategy profile of $(G,(s_1,t_1), \ldots, (s_n,t_n),c,d)$ is enforceable, follows very easily from the proof of statement (1). Note that, if $P$ is not enforceable, Algorithm~\nameref{alg} computes a strategy profile with strictly smaller cost. Since this would lead to a contradiction if $P$ is an optimal, but not enforceable strategy profile, every optimal strategy profile has to be enforceable. 

We finally want to show the last statement of \autoref{theo_sepa}, i.e. that an optimal Steiner forest can be computed in polynomial time (for $d_{i,e}=0$ for all $i\in N, e \in E$). 
To this end we want to use the result of Bateni et al. \cite{bateni2011} that the Steiner forest problem can be solved in polynomial time on graphs with treewidth at most 2. 
Thus it is sufficient to show that $G$ has treewidth at most 2, or, since generalized series-parallel graphs have treewidth at most 2 (what can easily be seen by induction on the number of operations), that $G$ is generalized series-parallel (note that we can assume w.l.o.g. that $G$ is connected, otherwise we can obviously treat each connected component separately).
Recall that generalized series-parallel graphs are created by a sequence of series, parallel, and/or add operations starting from a single edge, where an add-operation adds a new vertex $w$ and connects it to a given vertex $v$ by the edge $w-v$. 
We show that $G$ can be created like this. 
It is clear that this holds for each $G_i$ since they are series-parallel; but since the $G_i$s are (in general) neither equal nor disjoint, it is not completely obvious that this also holds for their union $G$. 

We now show, starting with the subgraph $G_1$ which is generalized series-parallel, that we can consecutively choose one player and add the vertices and edges of her paths which are not already contained in the subgraph constructed so far by add, series, and parallel operations. Since this again yields a generalized series-parallel graph, we finally  conclude that $G$ is generalized series-parallel. 

Let $G'\neq G$ be the generalized series-parallel subgraph constructed so far. Choose a player $i$ so that $G_i$ is not node-disjoint with $G'$ (exists since $G$ is connected).
Let $P_i$ be an $(s_i,t_i)$-path which is not node-disjoint with $G'$ and subdivide $P_i$ into the following three subpaths $P_i^1, P_i^2, P_i^3$ (where some of the subpaths may consist of only one node): $P_i^1$ starts in $s_i$ and ends in the first node $u$ which is contained in $G'$; $P_i^2$ starts in $u$ and ends in the last node $v$ which is in $G''$, and $P_i^3$ starts in $v$ and ends in $t_i$. Note that $G_i$ consists of $P_i$ together with all alternatives of player $i$ (according to $P_i$).
The following points show that $G_i \setminus G'$ can be added to $G'$ by series, parallel and add operations:
\begin{enumerate}
		\item $P_i^1$ ($P_i^3$) can obviously be added by an add operation at $u$ ($v$) and series operations.
		\item $P_i^2$ is completely contained in $G''$. Thus $P_i^2$ does not need to be added.
		\item Any alternatives where both endnodes are in $P_i^1$ or $P_i^3$ are internal node-disjoint with $G''$ and can therefore be added by parallel and series operations during the addition of $P_i^1$ and $P_i^3$. 
		\item Alternatives with both endnodes in $P_i^2$ are already contained in $G''$.
		\item There are no alternatives with endnodes in different subpaths.
\end{enumerate}
Note that 2.-5. holds since otherwise there would be a new $(s_j,t_j)$-path for a player $j$ already added; contradiction.

This completes the proof of statement (3). Hence, \autoref{theo_sepa} is shown.
\end{proof}

\begin{remark}
The first two results of~\autoref{theo_sepa} can be generalized to nonnegative, nondecreasing and discrete-concave shareable edge cost functions.
However, we do not know whether or not polynomial running time can be guaranteed.

%
\end{remark}
We now demonstrate that the assumption of $n$-sepa graphs is in some sense well justified.
\begin{theorem}\label{theo_gensepa}
For $n \geq 3$ players, there is a generalized series-parallel graph with fixed edge costs and no player-specific delays, so that the unique optimal Steiner forest is not enforceable. Therefore, a black-box reduction as for $n$-series-parallel graphs is impossible for generalized series-parallel graphs (even without player-specific delays). 
\end{theorem}

\begin{proof}
To prove \autoref{theo_gensepa}, consider Figure~\ref{counterexample_3p_3}. The displayed graph $G$ is generalized series-parallel since it can be created from a $K_2$ by a sequence of series- and parallel-operations (as executed in Figure~\ref{counterexample_3p_2}).
But the unique optimal Steiner forest OPT of $(G,(s_1,t_1),(s_2,t_2), (s_3,t_3),c)$, given by the solid edges, is not enforceable. To see this, note that the cost of OPT is $C(\text{OPT})=346$. Furthermore, we can upper-bound the sum of cost shares that the players will pay for using their paths in OPT by $100+69+170=339<C(\text{OPT})$, thus showing that OPT is not enforceable: Player 1 will pay at most 100, because she can use the edge $s_1-t_1$ with cost 100. Player 3 could use the edge $s_3-t_3$ with cost 69, thus she will pay at most 69. It remains to analyze the cost shares of Player 2. Instead of using the subpath from $s_2$ to $s_1$ of her path  in OPT, Player 2 could use the edge $s_2-s_1$ with cost 84. Furthermore, she could use the edge $t_3-t_2$ with cost 86 instead of her subpath from $t_3$ to $t_2$. Since the mentioned subpaths cover the complete path of Player 2 in OPT, she will pay at most $84+86=170$.
\\
For $n \geq 4$, we obviously get an instance with the properties stated in \autoref{theo_gensepa} by choosing an arbitrary node $v$ of $G$ and setting $s_i=t_i=v$ for all $i \in \{4,\ldots,n\}$. 
\end{proof}
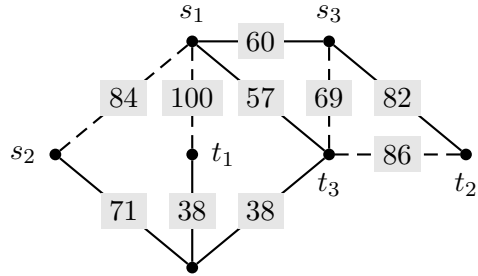
\begin{figure}[ht] \centering \psset{unit=0.6cm}
 \begin{pspicture}(0,0)(10,5) 
\Knoten{3}{5}{s1}\nput{90}{s1}{$s_1$}
\Knoten{3}{2.5}{t1}\nput{0}{t1}{$t_1$}
\Knoten{0}{2.5}{s2}\nput{180}{s2}{$s_2$}
\Knoten{9}{2.5}{t2}\nput{270}{t2}{$t_2$}
\Knoten{6}{5}{s3}\nput{90}{s3}{$s_3$}
\Knoten{6}{2.5}{t3}\nput{270}{t3}{$t_3$}
\Knoten{3}{0}{a}

\Kantedashed{s1}{s2}{84}
\Kantedashed{s1}{t1}{100}
\Kantedashed{t3}{s3}{69}
\Kantedashed{t2}{t3}{86}

\Kante{s1}{s3}{60}
\Kante{s1}{t3}{57}
\Kante{a}{s2}{71}
\Kante{a}{t1}{38}
\Kante{a}{t3}{38}
\Kante{t2}{s3}{82}
\end{pspicture}
\caption{Graph $G$ with three source-terminal pairs $(s_1,t_1),(s_2,t_2),(s_3,t_3)$, fixed edge costs $c$ given on the edges, and no player-specific delays $d$.}\label{counterexample_3p_3}
\end{figure}  

\begin{figure}[H] \centering \psset{unit=0.7cm} 
\begin{pspicture}(0,-0.8)(10,4.8) 
\Knoten{0}{5}{k1}
\Knoten{0}{3}{k2}
\ncline{k1}{k2}
\Knoten{8}{1.5}{s1}
\Knoten{8}{0.5}{t1}
\Knoten{7}{0.5}{s2}
\Knoten{10}{0.5}{t2}
\Knoten{9}{1.5}{s3}
\Knoten{9}{0.5}{t3}
\Knoten{8}{-0.5}{a}
\ncline{s1}{s2}
\ncline{s1}{t1}
\ncline{s1}{s3}
\ncline{s1}{t3}
\ncline{a}{s2}
\ncline{a}{t1}
\ncline{a}{t3}
\ncline{t3}{s3}
\ncline{t2}{s3}
\ncline{t2}{t3}

\Knoten{2}{5}{k3}
\Knoten{2}{3}{k4}
\ncline{k3}{k4}
\ncarc[arcangle=-60]{k3}{k4}
\ncarc[arcangle=60]{k3}{k4}

\Knoten{5}{5}{k5}
\Knoten{5}{3}{k6}
\Knoten{4}{4}{k7}
\Knoten{5}{4}{k8}
\Knoten{6}{4}{k9}
\ncline{k5}{k7}
\ncline{k7}{k6}
\ncline{k5}{k8}
\ncline{k8}{k6}
\ncline{k5}{k9}
\ncline{k9}{k6}

\Knoten{8}{5}{k10}
\Knoten{8}{3}{k11}
\Knoten{7}{4}{k12}
\Knoten{8}{4}{k13}
\Knoten{9}{4}{k14}
\ncline{k10}{k12}
\ncline{k12}{k11}
\ncline{k10}{k13}
\ncline{k13}{k11}
\ncline{k10}{k14}
\ncline{k14}{k11}
\ncarc[arcangle=60]{k10}{k14}

\Knoten{1}{1.5}{k15}
\Knoten{1}{-0.5}{k16}
\Knoten{0}{0.5}{k17}
\Knoten{1}{0.5}{k18}
\Knoten{2}{0.5}{k19}
\Knoten{2}{1.5}{k20}
\ncline{k15}{k17}
\ncline{k17}{k16}
\ncline{k15}{k18}
\ncline{k18}{k16}
\ncline{k15}{k19}
\ncline{k19}{k16}
\ncline{k15}{k20}
\ncline{k19}{k20}

\Knoten{4.5}{1.5}{k21}
\Knoten{4.5}{-0.5}{k22}
\Knoten{3.5}{0.5}{k23}
\Knoten{4.5}{0.5}{k24}
\Knoten{5.5}{0.5}{k25}
\Knoten{5.5}{1.5}{k26}
\ncline{k21}{k23}
\ncline{k23}{k22}
\ncline{k21}{k24}
\ncline{k24}{k22}
\ncline{k21}{k25}
\ncline{k25}{k22}
\ncline{k21}{k26}
\ncline{k25}{k26}
\ncarc[arcangle=60]{k26}{k25}
\end{pspicture}
\caption{Verification that $G$ is generalized series-parallel.}\label{counterexample_3p_2}
\end{figure}
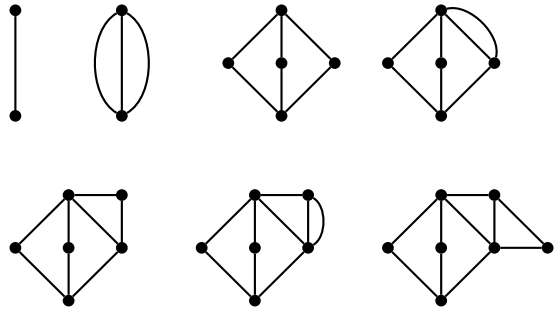




\begin{algorithm}[H]
\caption{\textsc{Irredundant}} \label{irredundant}
\KwData{Undirected graph $G$, source-terminal pairs $(s_1,t_1),\ldots,(s_n,t_n)$}
\KwResult{Maximum irredundant subgraph $G'$}
$C \leftarrow$ set of cut vertices of $G$; \\
\For{each $i \in \{1,\ldots,n\}$}{
$G_i\leftarrow G$; \\
\For{each $c \in C$}{
Remove from $G_i$ all vertices in the components of $G-c$ containing neither $s_i$ nor $t_i$ (if any)
}}
$G'\leftarrow G$; \\
Delete from $G'$ all nodes and edges which are not contained in any $G_i$.
\end{algorithm} 

{ \ }

\begin{algorithm}[H]
\caption{$n$-\textsc{SePa}} \label{alg}
\KwData{Connection Game $(N, G,(s_1,t_1),\ldots,(s_n,t_n),c,d)$ with $n$-series-parallel graph $(G,(s_1,t_1),\ldots,(s_n,t_n))$; strategy profile $P=(P_1,\ldots,P_n)$}
\KwResult{Enforceable strategy profile $P'$ with cost $C(P') \leq C(P)$; cost share functions $\xi$ so that $P'$ is PNE}
Solve LP($P$); let $(\xi_{i,e})_{i \in N, e \in P_i}$ be the computed optimal solution; \\
$P' \leftarrow P$; \\
$\xi_{i,e}(P') \leftarrow \xi_{i,e}$ for all $i \in N$, $e \in P_i'$; \\
\uIf{(\ref{BB}) does not hold}{
Let $S$ be the set of edges which are not completely paid according to $P'$; \\
\Repeat{$S=\emptyset$}{
All players with edges in $S$ deviate from all unpaid edges by using a combination of smallest tight alternatives; \\
Update $P'$ accordingly; \\
Update the cost shares (for each player $i$ and each edge $e \in P_i'$): 
	\[
	\xi_{i,e}(P')=\begin{cases}
	\xi_{i,e}, & \text{for $e \in P_i' \cap P_i$,} \\
	c_e, & \text{for $e \in P_i' \setminus P_i$.}
	\end{cases}
\]
Update $S$ (set of edges which are not completely paid according to $P'$);
}
If there are overpaid edges, decrease the corresponding cost shares arbitrarily until these edges are exactly paid; \\
}
\textbf{output} $P'$ and $\xi$ (induced by $(\xi_{i,e}(P'))_{i \in N, e \in P_i'}$);
\end{algorithm} 

{ \ }

\begin{algorithm}[H]
\caption{\textsc{Alternatives($i$)}} \label{algpaths}
\KwData{Connection Game $(N, G,(s_1,t_1),\ldots,(s_n,t_n),c,d)$ with $n$-series-parallel graph $(G,(s_1,t_1),\ldots,(s_n,t_n))$; player $i \in N$; path $P_i \in \mathcal P_i$}
\KwResult{Set of cheapest alternatives $\mathcal A_i$ according to $P_i$}
Define $\tilde{c}_e:=c_e+d_{i,e}$ for all $e \in G_i$; \\
Delete all edges of $P_i$; \\
\Repeat{all nodes of $P_i$ have been deleted}{
Let $u$ be the first node of $P_i$ which is not yet deleted; \\
\Repeat{$u$ has been deleted}{
Starting in $u$, execute a (partial) BFS in $G_i$, stop if a new node of $P_i$ is reached; \\
\uIf{a new node $v$ of $P_i$ is reached}{
Compute a shortest $u-v$-path $A_i$ (shortest according to $\tilde{c}$); \\
Insert $A_i$ in $\mathcal A_i$; \\
Delete all nodes (different from $u,v$) and edges of $A_i$ from $G_i$; \\
Mark $v$ as \textit{visited};
}
\Else{Delete all nodes ($\notin P_i$) and edges of the connected component of $u$; \\
Delete $u$; \\
Remove all marks.
}
}
}
\end{algorithm} 


\bibliography{master-bib,literature}  

\begin{thebibliography}{10}

\bibitem{AndelmanFM09}
N. Andelman, M. Feldman, and Y. Mansour.
\newblock Strong price of anarchy.
\newblock {\em Games Econom.\ Behav.}, 65(2):289--317, 2009.

\bibitem{AnshelevichDTW08}
E. Anshelevich, A. Dasgupta, {\'E}. Tardos, and T. Wexler.
\newblock Near-optimal network design with selfish agents.
\newblock {\em Theory of Computing}, 4(1):77--109, 2008.

\bibitem{AnshelevichC09b}
E. Anshelevich and B. Caskurlu.
\newblock Exact and approximate equilibria for optimal group network formation.
\newblock {\em Theor.\ Comput.\ Sci.}, 412(39):5298--5314, 2011.

\bibitem{AnshelevichC09a}
E. Anshelevich and B. Caskurlu.
\newblock Price of stability in survivable network design.
\newblock {\em Theory Comput.\ Syst.}, 49(1):98--138, 2011.

\bibitem{AnshelevichDKRTW08}
E. Anshelevich, A. Dasgupta, J. Kleinberg, T. Roughgarden, {\'E}. Tardos, and
  T. Wexler.
\newblock The price of stability for network design with fair cost allocation.
\newblock {\em SIAM J. Comput.}, 38(4):1602--1623, 2008.

\bibitem{AvniT16}
G. Avni and T. Tamir.
\newblock Cost-sharing scheduling games on restricted unrelated machines.
\newblock {\em Theor. Comput. Sci.}, 646:26--39, 2016.

\bibitem{bateni2011}
M. Bateni, M. Hajiaghayi, and D. Marx.
\newblock Approximation schemes for steiner forest on planar graphs and graphs
  of bounded treewidth.
\newblock {\em JACM}, 58(5):21:1--21:37, 2011.

\bibitem{Bilo2010}
V. Bil\'{o}, A. Fanelli, M. Flammini, and L. Moscardelli.
\newblock When ignorance helps: Graphical multicast cost sharing games.
\newblock {\em Theoretical Computer Science}, 411(3):660 -- 671, 2010.

\bibitem{BiloFM13}
V. Bil{\`{o}}, M. Flammini, and L. Moscardelli.
\newblock The price of stability for undirected broadcast network design with
  fair cost allocation is constant.
\newblock In {\em Proc.\ 54th Symp.\ Foundations of Computer Science (FOCS)},
  pages 638--647, 2013.

\bibitem{Bird76}
C. Bird.
\newblock On cost allocation for a spanning tree: {A} game theoretic approach.
\newblock {\em Networks}, 6:335--350, 1976.

\bibitem{Byrka10}
J. Byrka and K. Aardal.
\newblock An optimal bifactor approximation algorithm for the metric
  uncapacitated facility location problem.
\newblock {\em SIAM Journal on Computing}, 39(6):2212--2231, 2010.

\bibitem{ByrkaGRS13}
J. Byrka, F. Grandoni, T. Rothvo{\ss}, and L. Sanit{\`{a}}.
\newblock Steiner tree approximation via iterative randomized rounding.
\newblock {\em J. ACM}, 60(1):6:1--6:33, 2013.

\bibitem{CGV17}
I. Caragiannis, V. Gkatzelis, and C. Vinci.
\newblock Coordination mechanisms, cost-sharing, and approximation algorithms
  for scheduling.
\newblock In N. R.~Devanur and P. Lu, editors, {\em Web and Internet
  Economics}, pages 74--87, 2017.

\bibitem{CardinalH10}
J. Cardinal and M. Hoefer.
\newblock Non-cooperative facility location and covering games.
\newblock {\em Theor. Comput. Sci.}, 411:1855--1876, March 2010.

\bibitem{CharikarCCDGG99}
M. Charikar, C. Chekuri, T.-Y. Cheung, Z. Dai, A. Goel, and S. Guha.
\newblock Approximation algorithms for directed {S}teiner problems.
\newblock {\em J. Algorithms}, 33(1):192--200, 1999.

\bibitem{ChekuriP05}
C. Chekuri and M. P{\'{a}}l.
\newblock A recursive greedy algorithm for walks in directed graphs.
\newblock In {\em Proc.\ 46th Symp.\ Foundations of Computer Science (FOCS)},
  pages 245--253, 2005.

\bibitem{ChenRV10}
H.-L. Chen, T. Roughgarden, and G. Valiant.
\newblock Designing network protocols for good equilibria.
\newblock {\em SIAM J. Comput.}, 39(5):1799--1832, 2010.

\bibitem{chen2016}
X. Chen, Z. Diao, and X. Hu.
\newblock Network characterizations for excluding {Braess's} paradox.
\newblock {\em Theory Comput.\ Syst.}, 59(4):747--780, 2016.

\bibitem{ChristodoulouGS17}
G. Christodoulou, V. Gkatzelis, and A. Sgouritsa.
\newblock Cost-sharing methods for scheduling games under uncertainty.
\newblock In {\em Proc.\ 18th {ACM} Conf.\ Economics and Computation (EC)},
  pages 441--458, 2017.

\bibitem{LS16}
G. Christodoulou, S. Leonardi, and A. Sgouritsa.
\newblock Designing cost-sharing methods for bayesian games.
\newblock In {\em Proc.\ 9th Intl.\ Symp.\ Algorithmic Game Theory (SAGT)},
  pages 327--339, 2016.

\bibitem{CS16}
G. Christodoulou and A. Sgouritsa.
\newblock Designing networks with good equilibria under uncertainty.
\newblock In {\em Proc.\ 27th Symp. Discrete Algorithms (SODA)}, pages 72--89,
  2016.

\bibitem{DengIN99}
X. Deng, T. Ibaraki, and H. Nagamochi.
\newblock Algorithmic aspects of the core of combinatorial optimization games.
\newblock {\em Math.\ Oper.\ Res.}, 24(3):751--766, 1999.

\bibitem{DisserFKM15}
Y. Disser, A. Feldmann, M. Klimm, and M. Mihal{\'{a}}k.
\newblock Improving the {$H_k$}-bound on the price of stability in undirected
  shapley network design games.
\newblock {\em Theoret.\ Comput.\ Sci.}, 562:557--564, 2015.

\bibitem{Feldman12}
M. Feldman and T. Tamir.
\newblock Conflicting congestion effects in resource allocation games.
\newblock {\em Oper. Res.}, 60(3):529--540, 2012.

\bibitem{FiatKLOS06}
A. Fiat, H. Kaplan, M. Levy, S. Olonetzky, and R. Shabo.
\newblock On the price of stability for designing undirected networks with fair
  cost allocations.
\newblock In {\em Proc.\ 33rd Intl.\ Coll.\ Automata, Languages and Programming
  (ICALP)}, volume~1, pages 608--618, 2006.

\bibitem{GargKR00}
N. Garg, G. Konjevod, and R. Ravi.
\newblock A polylogarithmic approximation algorithm for the {G}roup {S}teiner
  tree problem.
\newblock {\em J. Algorithms}, 37:66--84, 2000.

\bibitem{GeorgiouS12}
K. Georgiou and C. Swamy.
\newblock Black-box reductions for cost-sharing mechanism design.
\newblock In {\em Proc.\ 23rd Symp.\ Discrete Algorithms (SODA)}, pages
  896--913, 2012.

\bibitem{GoemansS04}
M.~X. Goemans and M. Skutella.
\newblock Cooperative facility location games.
\newblock {\em J. Algorithms}, 50(2):194--214, 2004.

\bibitem{GranotH81}
D. Granot and G. Huberman.
\newblock On minimum cost spanning tree games.
\newblock {\em Math.\ Prog.}, 21:1--18, 1981.

\bibitem{GranotM98}
D. Granot and M. Maschler.
\newblock Spanning network games.
\newblock {\em Int.\ J. Game Theory}, 27:467--500, 1998.

\bibitem{GuhaK99}
S. Guha and S. Khuller.
\newblock Greedy strikes back: {I}mproved facility location algorithms.
\newblock {\em J. Algorithms}, 31:228--248, 1999.

\bibitem{GuptaK15}
A. Gupta and A. Kumar.
\newblock Greedy algorithms for {S}teiner forest.
\newblock In {\em Proc.\ 47th Symp.\ Theory of Computing (STOC)}, pages
  871--878, 2015.

\bibitem{HansenT09}
T.~D. Hansen and O. Telelis.
\newblock Improved bounds for facility location games with fair cost
  allocation.
\newblock In {\em Proc.\ 3rd Intl.\ Conf.\ Combinatorial Optimization and
  Applications (COCOA)}, pages 174--185, 2009.

\bibitem{harks2017}
T. Harks, A. Huber, and M. Surek.
\newblock A characterization of undirected graphs admitting optimal cost
  shares.
\newblock In N. R.~Devanur and P. Lu, editors, {\em Web and Internet
  Economics}, pages 237--251, Cham, 2017. Springer International Publishing.

\bibitem{HarksP14}
T. Harks and B. Peis.
\newblock Resource buying games.
\newblock {\em Algorithmica}, 70(3):493--512, 2014.

\bibitem{HarksF14}
T. Harks and P. von Falkenhausen.
\newblock Optimal cost sharing for capacitated facility location games.
\newblock {\em European Journal of Operational Research}, 239(1):187--198,
  2014.

\bibitem{Hoefer09}
M. Hoefer.
\newblock Non-cooperative tree creation.
\newblock {\em Algorithmica}, 53:104--131, 2009.

\bibitem{Hoefer11}
M. Hoefer.
\newblock Competitive cost sharing with economies of scale.
\newblock {\em Algorithmica}, 60:743--765, 2011.

\bibitem{Hoefer13}
M. Hoefer.
\newblock Strategic cooperation in cost sharing games.
\newblock {\em Internat. J. Game Theory}, 42(1):29--53, 2013.

\bibitem{JainV01}
K. Jain and V. Vazirani.
\newblock Applications of approximation algorithms to cooperative games.
\newblock In {\em Proc.\ 33rd Symp.\ Theory of Computing (STOC)}, pages
  364--372, 2001.

\bibitem{KonemannLSZ08}
J. K{\"o}nemann, S. Leonardi, G. Sch{\"a}fer, and S. van Zwam.
\newblock A group-strategyproof cost sharing mechanism for the steiner forest
  game.
\newblock {\em SIAM J. Comput.}, 37(5):1319--1341, 2008.

\bibitem{LeeL13}
E. Lee and K. Ligett.
\newblock Improved bounds on the price of stability in network cost sharing
  games.
\newblock In {\em Proc.\ 14th Conf.\ Electronic Commerce (EC)}, pages 607--620,
  2013.

\bibitem{Li09}
J. Li.
\newblock An o(log(n)/log(log(n))) upper bound on the price of stability for
  undirected shapley network design games.
\newblock {\em Inf.\ Process.\ Lett.}, 109(15):876--878, 2009.

\bibitem{Li:2013}
S. Li.
\newblock A 1.488 approximation algorithm for the uncapacitated facility
  location problem.
\newblock {\em Inf. Comput.}, 222:45--58, 2013.

\bibitem{Megiddo78}
N. Megiddo.
\newblock Cost allocation for {S}teiner trees.
\newblock {\em Networks}, 8(1):1--6, 1978.

\bibitem{MoulinShenker}
H. Moulin and S. Shenker.
\newblock Strategyproof sharing of submodular costs: budget balance versus
  efficiency.
\newblock {\em Econom. Theory}, 18(3):511--533, 2001.

\bibitem{PalT03}
M. P{\'a}l and {\'E}. Tardos.
\newblock Group strategyproof mechanisms via primal-dual algorithms.
\newblock In {\em FOCS}, pages 584--593, 2003.

\bibitem{RobinsZ05}
G. Robins and A. Zelikovsky.
\newblock Tighter bounds for graph {S}teiner tree approximation.
\newblock {\em SIAM J. Disc.\ Math.}, 19(1):122--134, 2005.

\bibitem{Rosenthal73}
R. Rosenthal.
\newblock A class of games possessing pure-strategy {N}ash equilibria.
\newblock {\em Int.\ J. Game Theory}, 2:65--67, 1973.

\bibitem{Syrgkanis10}
V. Syrgkanis.
\newblock The complexity of equilibria in cost sharing games.
\newblock In A. Saberi, editor, {\em Proc. 6th Internat. Workshop on Internet
  and Network Econom.}, LNCS, pages 366--377, 2010.

\bibitem{Tamir91}
A. Tamir.
\newblock On the core of network synthesis games.
\newblock {\em Math.\ Prog.}, 50:123--135, 1991.

\bibitem{HarksF13}
P. von Falkenhausen and T. Harks.
\newblock Optimal cost sharing for resource selection games.
\newblock {\em Math.\ Oper.\ Res.}, 38(1):184--208, 2013.

\end{thebibliography}
\end{document}